\documentclass[submission,copyright,creativecommons]{eptcs}
\providecommand{\event}{EXPRESS/SOS 2023} 

\usepackage{iftex}

\ifpdf
  \usepackage{underscore}         
  \usepackage[T1]{fontenc}        
\else
  \usepackage{breakurl}           
\fi

\usepackage{microtype}
\usepackage{todonotes}
\usepackage{amsmath,amsfonts,amssymb,amsthm} 
\usepackage{cases}
\usepackage{bm}
\usepackage{stmaryrd}
\usepackage{hyperref}
\usepackage{enumerate}
\usepackage{wrapfig}
\usepackage{multirow}
\usepackage{verbatim}
\usepackage{mathtools}
\usepackage{booktabs}

\usepackage{paralist}
\usepackage{extarrows}

\usepackage{wasysym}

\usepackage{algorithm}
\usepackage{algorithmic}
\usepackage{listings}
\usepackage[title]{appendix}
\usepackage{marginnote}
\usepackage{graphicx}
\usepackage{xcolor,colortbl}
\usepackage{pifont}
\usepackage{hhline}
\usepackage{rotating}
\usepackage{caption}
\usepackage{subcaption}
\allowdisplaybreaks

\usepackage{pgfplots}
\usepackage{tikz}
\usetikzlibrary{arrows,arrows.meta,calc,automata,positioning,decorations.pathreplacing,shapes.geometric,shapes.misc,graphs,backgrounds,shadows.blur}
\tikzset{align at top/.style={baseline=(current bounding box.north)}}
\tikzstyle{every node}=[font=\scriptsize]
\tikzstyle{state} = [draw,fill=white,ellipse,thick,align=center,inner sep=0pt,minimum size=4.5mm]
\tikzstyle{vvert} = [draw,fill=white,ellipse,thick,align=center,inner sep=-2pt,minimum size=8mm]
\tikzstyle{rvert} = [draw,fill=white,rectangle,thick,align=center,inner sep=3pt,minimum size=7mm]
\tikzstyle{dot} = [fill,circle,inner sep=0mm,minimum size=1.25mm,line width=0mm]
\newtheorem{definition}{Definition}
\newtheorem{theorem}{Theorem}
\newtheorem{example}{Example}
\newtheorem{lemma}{Lemma}
\newtheorem{proposition}{Proposition}

\newcommand{\Nat}{\mathbb{N}}
\newcommand{\Reals}{\mathbb{R}}

\newcommand{\M}{\mathbin{\mathbf{M}}} 
\def\Masking{\preceq_{m}}

\def\FMask{f_{\milestones}}

\def\Refuter{\mathsf{R}}
\def\Verifier{\mathsf{V}}
\def\Probabilistic{\mathsf{P}}
\def\ErrorSt{{v_{\text{err}}}}

\def\InitVertex{v_0^\StochG}
\def\InitVertexH{v_0^\StochH}
\def\InitVertexSG{v_0^{\SymbG}}

\def\SigmaF{\Sigma_{\mathcal{F}}}

\def\val{\mathop{\textup{val}}}
\def\support#1{\mathit{Supp}\left(#1\right)}

\def\faults{\mathcal{F}}
\def\RFP{\mathit{RFP}}
\def\Prob#1#2{\mathit{Prob}^{#1, #2}}
\def\couplings#1#2{\mathbb{C}(#1,#2)}
\def\verticesletter{\mathbb{V}}
\def\vertices#1{\verticesletter(#1)}

\def\pr#1#2{{#2}[{#1}]}

\def\RelCoupling{R^{\#}}
\newcommand{\MaskCoup}{\mathbin{\mathbf{M^{\#}}}} 
\def\Expect#1#2{\mathbb{E}^{#1, #2}}
\def\SymbG{\mathcal{SG}}
\def\StochG{\mathcal{G}}
\def\StochH{\mathcal{H}}

\def\post{\mathit{Post}}
\def\pre{\mathit{Pre}}

\def\SymbAFairpre{{\forall\pre^\SymbG_\fairidx}}
\def\SymbEFairpre{{\exists\pre^\SymbG_\fairidx}}
\def\Eq{\mathit{Eq}}
\def\out{\mathit{out}}
\def\Dist{\mathcal{D}}
\def\Dirac{\Delta}
\def\reward{\mathit{r}}

\def\tuple#1{{\langle{#1}\rangle}}

\def\pparenthesis#1{{\llparenthesis #1 \rrparenthesis}}

\def\strat#1{\pi_{#1}}

\def\starredstrat#1{\pi^*_{#1}}

\def\upperbound{\mathbf{U}}
\def\fairidx{\textsl{f}}
\def\memorylessidx{\textsl{M}}
\def\deterministicidx{\textsl{D}}
\def\semimarkovidx{\textsl{S}}
\def\extremeidx{\textsl{X}}
\def\Strategies#1{\Pi_{#1}}

\def\FairStrats#1{\Pi^{\fairidx}_{#1}}

\def\SemiMarkovFairStrats#1{\Pi^{\semimarkovidx\fairidx}_{#1}}
\def\SemiMarkovStrats#1{\Pi^{\semimarkovidx}_{#1}}
\def\XSemiMarkovStrats#1{\Pi^{\extremeidx\semimarkovidx}_{#1}}
\def\DetMemorylessStrats#1{\Pi^{\memorylessidx\deterministicidx}_{#1}}
\def\XDetMemorylessStrats#1{\Pi^{\extremeidx\memorylessidx\deterministicidx}_{#1}}
\def\DetMemorylessFairStrats#1{\Pi^{\memorylessidx\deterministicidx\fairidx}_{#1}}

\newcommand{\milestones}{\mathsf{m}}
\newcommand{\mreward}{\reward_\milestones}

\mathchardef\mhyphen="2D

\newcommand*{\xrightarrowprime}[2][]{\mathrel{{\xrightarrow[#1]{#2}}{}'}}
\def\PRISM{\textsf{PRISM}}
\def\LTL{\textsf{LTL}}
\def\Bellman{\Gamma}

\newcommand{\codett}[1]{\ensuremath{\mathtt{#1}}}

\usepackage{color}

\definecolor{lightblue}{RGB}{220,220,255}
\definecolor{lightred}{RGB}{255,224,224}
\definecolor{lightgreen}{RGB}{224,255,224}
\definecolor{lightyellow}{RGB}{255,255,224}
\definecolor{lightpurple}{RGB}{255,224,255}
\definecolor{darkerred}{RGB}{64,0,0}
\definecolor{darkred}{RGB}{128,0,0}
\definecolor{darkblue}{RGB}{0,0,128}
\definecolor{darkgreen}{RGB}{0,128,0}
\definecolor{darkpurple}{RGB}{128,0,128}
\definecolor{black}{RGB}{0,0,0}

\makeatletter
\def\THICKhrulefill{\leavevmode \leaders \hrule height 5pt\hfill \kern \z@}
\makeatother

\urldef{\mailsa}\path|{pedro.dargenio,rdemasi}@unc.edu.ar|
\urldef{\mailsb}\path|{pcastro,lputruele}@dc.exa.unrc.edu.ar|

\graphicspath{{./Figs/}}

\makeatletter
\renewcommand\paragraph{\@startsection{paragraph}{4}{\z@}%
                       {-5\p@ \@plus -2\p@ \@minus -2\p@}%
                       {-0.5em \@plus -0.22em \@minus -0.1em}%
                       {\normalfont\normalsize\itshape}}
\makeatother

\title{Quantifying Masking Fault-Tolerance\\ via Fair Stochastic Games%
  \thanks{This work was supported by ANPCyT PICT-2017-3894
    (RAFTSys), ANPCyT PICT 2019-3134, SeCyT-UNC 33620180100354CB
    (ARES), and EU Horizon
    2020 MSCA grant agreement 101008233 (MISSION).}
}
\author{Pablo F. Castro
\institute{Departamento de Computaci\'on, FCEFQyN, Universidad Nacional de R\'{\i}o Cuarto, R\'io Cuarto, Argentina} 
\institute{Consejo Nacional de Investigaciones Cient\'ificas y T\'ecnicas (CONICET), Argentina }
\email{pcastro@dc.exa.unrc.edu.ar}
\and
 Pedro R. D'Argenio
\institute{FAMAF, Universidad Nacional de C\'ordoba, C\'ordoba, Argentina} 
\institute{Consejo Nacional de Investigaciones Cient\'ificas y T\'ecnicas (CONICET), Argentina}
\email{pedro.dargenio@unc.edu.ar}
\and
 Ramiro Demasi
\institute{FAMAF, Universidad Nacional de C\'ordoba, C\'ordoba, Argentina} 
\institute{Consejo Nacional de Investigaciones Cient\'ificas y T\'ecnicas (CONICET), Argentina}
\email{rdemasi@unc.edu.ar}
\and
 Luciano Putruele
\institute{Departamento de Computaci\'on, FCEFQyN, Universidad Nacional de R\'{\i}o Cuarto, R\'io Cuarto, Argentina}
\institute{Consejo Nacional de Investigaciones Cient\'ificas y T\'ecnicas (CONICET), Argentina}
\email{lputruele@dc.exa.unrc.edu.ar}
}
\def\titlerunning{Quantifying Masking Fault-Tolerance via ...}
\def\authorrunning{Castro,  D'Argenio,  Demasi \& Putruele}

\begin{document}

\maketitle

\begin{abstract}
  We introduce a formal notion of masking fault-tolerance between
  probabilistic transition systems using stochastic games. These games
  are inspired in bisimulation games,  but  they also take into account the possible faulty behavior of systems.  
  When no faults are present,  these games boil down to probabilistic bisimulation games.
 Since these games could be infinite, we propose a symbolic way of
  representing them so that they can be solved in polynomial time.
  In particular,  we use this notion of masking to quantify the level of masking
  fault-tolerance exhibited by almost-sure failing systems, i.e.,
  those systems that eventually fail with probability $1$.  The level
  of masking fault-tolerance of almost-sure failing systems can be
  calculated by solving a collection of functional equations.
  We produce this metric in a setting in which one of the player 
  behaves in a strong fair way (mimicking the idea of fair environments).
\end{abstract}

\section{Introduction} \label{sec:intro}

Fault-tolerance~\cite{Gartner99} is an important aspect of critical systems,  in which a fault may lead to important economic, or human life,  losses.  Examples are ubiquitous: banking systems, automotive software, communication protocols, etc.  Fault-tolerant systems  typically use some kind of mechanism based on redundancy such as data replication,  duplicated messages and voting.  However, these techniques do not consistently enhance the ability of systems to effectively tolerate faults as one could expect.  Hence,  quantifying the effectiveness of fault-tolerance mechanisms is an important issue when developing critical software.  Additionally,  in most cases, faults have a probabilistic nature, 
thus any technique designed for measuring system fault-tolerance should be able to cope with stochastic phenomena.    

 
 In this paper we
provide a framework aimed at quantifying the fault-tolerance exhibited by concurrent probabilistic systems.  This encompasses  the probability of occurrence of faults as well as the use
of randomized algorithms.  Particularly,  we focus on the so-called \emph{masking fault-tolerance},   in which both the safety and liveness properties are preserved by the system  under the occurrence of faults~\cite{Gartner99}. Intuitively,   faults are masked in such a way that their occurrence cannot be observed by the users.  This is often acknowledged as the most desirable kind of fault-tolerance.  The aim of this paper is to provide a framework for selecting a fault-tolerance mechanism over others as well as for balancing multiple mechanisms
(e.g.,  to ponder on cost efficient hardware redundancies vs.\ time demanding software artifacts).

	
 In the last years,  significant progress has been made towards 
defining suitable metrics or distances for diverse types of quantitative 
models including real-time systems \cite{HenzingerMP05}, probabilistic 
models \cite{DBLP:conf/ifip2/GiacaloneJS90,DesharnaisGJP04,BreugelW06,DesharnaisLT11,DBLP:conf/birthday/BreugelW14,Bacci0LM17,TangB18,BacciBLMTB19}, 
and metrics for linear and branching systems \cite{AlfaroFS09,ThraneFL10,LarsenFT11,CernyHR12,Henzinger13}. 
Some authors have already pointed out that these metrics can be useful to reason 
about the robustness and correctness of a system, notions related to fault-tolerance. 
Here we follow the ideas introduced in \cite{CastroDDP18b} where  masking fault-tolerance is captured by means of a tailored  bisimulation
game with quantitative objectives.  We extend these ideas to a probabilistic setting and define a probabilistic version of 
this characterization of masking fault-tolerance which, in turn, we use to define a metric to compare the ``degree'' of masking fault tolerance provided by different mechanisms.


More specifically, we characterize probabilistic masking
fault-tolerance via a tailored variant of probabilistic bisimulation
(named \emph{masking simulation}).  Roughly speaking, masking
simulation relates two probabilistic transition systems. One of them
acts as a system specification (i.e., a nominal model), while the
other one can be thought of as a fault-tolerant implementation that
takes into account possible faulty behavior. The existence of a
masking simulation implies that the implementation masks all faults.
This relation admits a simple game characterization via
a Boolean reachability game played on a stochastic game graph.

Since in practice masking
fault tolerance cannot be achieved in full, the reliability of a fault
tolerance mechanism can only be measured quantitatively.  Thus, we
reinterpret the same game with quantitative objectives.
While previously we dealt with a Boolean reachability objective,
here we introduce \emph{milestones} indicating successful progress of
the model and change the objective of the game to be the expected
total collected milestones.  Therefore, we transform the game into an
expected total reward game.  We then take the measure of the
fault-tolerant mechanism to be the solution of this expected total
reward game.

 In order to prove our results we have addressed several technical issues.  First, the games rely on the notion of
couplings between probabilistic distributions and, as a consequence,
the number of vertices of their game graphs is infinite.  To be able to deal with these infinite games, 
we introduce a symbolic representation for them
where couplings are
captured by means of equation systems.  The size of these symbolic
graphs is polynomial in the size of the input systems, which enables us to
solve the (Boolean) simulation game in polynomial time.

Besides, stochastic games with expected total reward objectives are
required to be almost surely stopping~\cite{FilarV96} or, more
generally, almost surely stopping under
fairness~\cite{DBLP:conf/cav/CastroDDP22}.  In our terms, this means that the game needs to be \emph{almost
surely failing under fairness}.
Intuitively, these games model systems that will eventually fail with
probability $1$.  This generalizes the idea that faults with
some positive probability of occurrence will eventually
occur during a long enough system execution.

As our game is of infinite nature,  the results
in~\cite{DBLP:conf/cav/CastroDDP22} cannot be applied directly.
Therefore we devise a finite discretization that allows us to partly
reuse~\cite{DBLP:conf/cav/CastroDDP22} and show that the value of the
game is determined and that it can be computed by solving a collection
of functional equations via an adapted \emph{value iteration}
algorithm \cite{Condon90,Condon92,ChatterjeeH08,KelmendiKKW18}.
%
%
Besides, as the game can only be solved if the game is almost surely
failing under fairness we also provide a polynomial solution to
solve this problem.
%
We remark that both checking almost surely stopping under fairness and
solving the game are calculated through the symbolic graph.

Summarizing,
  we define the notion of probabilistic masking simulation and
  provide its game characterization which we show 
  decidable in polynomial time
  (Sec.~\ref{sec:probMaskDist}).
In Sec.~\ref{sec:probAlmostSure}
  we define an extension of the games by considering rewards and provide
  a payoff function that collects the ``milestones'' achieved
  by the implementation.
  We show that these games are determined provided they are almost-surely
  failing under fairness, and
  give an algorithm to calculate the value of these games.
  We also give a polynomial time algorithm to decide if a game
  is almost-surely failing under fairness.

\section{Preliminaries}\label{sec:background}


A (discrete) \emph{probability distribution} $\mu$ over a denumerable 
set $S$ is a function $\mu: S \rightarrow [0, 1] $  such that 
$\mu(S) \triangleq \sum_{s \in S} \mu(s) = 1$. 
Let $\Dist(S)$ denote the set of all probability distributions on $S$. $\Dirac_s \in \Dist(S)$ denotes the Dirac distribution for $s$, i.e., 
$\Dirac_s(s) =1$ and $\Dirac_s(s') = 0$ whenever $s'\neq s$.
The \textit{support} set of $\mu$ is defined by $\support{\mu} = \{s \mid {\mu (s) > 0}\}$.

A \emph{Probabilistic Transition System} (PTS)~\cite{Segala95} is a structure $A =( S, \Sigma, \rightarrow, s_0 )$ 
where 
\begin{inparaenum}[(i)]
\item%
  $S$ is a denumerable set of \emph{states} containing the
  \emph{initial state} $s_0 \in S$,
\item%
  $\Sigma$ is a set of \emph{actions}, and
\item%
  ${\rightarrow} \subseteq S \times \Sigma \times \Dist(S)$ is the
  \emph{(probabilistic) transition relation}.
\end{inparaenum}
We assume that there is always some transition leaving from every
state.  Here, we only consider finite PTSs, i.e., those in which the
set of states $S$, the set of actions $\Sigma$ and the transition
relation $\rightarrow$ are finite.


A distribution $w \in \Dist(S \times S')$ is a \textit{coupling} for $(\mu, \mu')$, with 
$\mu \in \Dist(S)$ and $\mu' \in \Dist(S')$, if $w(S, \cdot) = \mu'$ and 
$w(\cdot,S') = \mu$.  $\couplings{\mu}{\mu'}$ denotes the set of all couplings for $(\mu, \mu')$.  It is worth
noting that this defines a (two-way transport) polytope (i.e.,  a particular kind of bounded polyhedron).  $\vertices{\couplings{\mu}{\mu'}}$ denotes the set of all vertices of the corresponding polytope.
This set is finite if $S$ and $S'$ are finite.
For $R \subseteq S \times S'$, we say that a coupling $w$ for $(\mu, \mu')$ 
\textit{respects} $R$ if $\support{w} \subseteq R$ (i.e., $w(s, s') > 0 \Rightarrow s~R~s'$).
We define $\RelCoupling \subseteq \Dist(S) \times \Dist(S')$ by $\mu~\RelCoupling~\mu'$ if and only if
there is an $R$-respecting coupling for $(\mu, \mu')$.

A \emph{stochastic game graph} \cite{ChatterjeeH12} is a tuple $\StochG = ( V, E, V_1, V_2, V_\Probabilistic, v_0, \delta  ) $, where $V$ is a set of vertices with $V_1, V_2, V_\Probabilistic \subseteq V$ being a partition of $V$, $v_0\in V$ is the initial vertex,  $E\subseteq V \times V$,
and $\delta : V_\Probabilistic  \rightarrow \Dist(V)$ is a probabilistic transition function such that,  for all $v \in V_\Probabilistic$ and $v' \in V$:  $(v,v') \in E$ iff $v' \in \support{\delta(v)}$.
$V_1$ and $V_2$ are the set of vertices where Players 1 and 2 are
respectively allowed to play.
If $V_\Probabilistic = \emptyset$, then $\StochG$ is called a $2$-player game graph. Moreover, if $V_1 = \emptyset$ or $V_2 = \emptyset$, then $\StochG$ is a \emph{Markov Decision Process} (or MDP). Finally, in case that $V_1= \emptyset$ and $V_2 = \emptyset$, $\StochG$ is a \emph{Markov chain} (or MC). For all states $v \in V$ we define $\post(v) = \{v' \in V \mid (v,v') \in E\}$, the set of successors of $v$. Similarly, we define $\pre(v') = \{v \in V \mid (v,v') \in E \}$ as the set of predecessors of $v'$. 
We assume that $\post(v) \neq \emptyset$ for every $v \in V_1 \cup V_2$. 

Given a game as defined above, a \emph{play} is an infinite sequence $\rho =  \rho_0, \rho_1, \dots$ such that $(\rho_k, \rho_{k+1}) \in E$ for every $k \in \mathbb{N}$.  The set of all plays is denoted by $\Omega$,
and the set of plays starting at vertex $v$ is written $\Omega_v$. A \emph{strategy} (or policy) for Player $i\in\{1,2\}$ is a function $\strat{i}: V^*  \cdot V_i \rightarrow \Dist(V)$ that assigns a probabilistic distribution to each finite sequence of states such that $\support{\strat{i}(\rho\cdot v)}\subseteq\post(v)$ for all $\rho \in V^*$ and $v\in V_i$.   The set of all the strategies for Player $i$ is named $\Strategies{i}$. A strategy $\strat{i}$ is said to be  \emph{pure} (or \emph{deterministic}) if, for every $\rho\in V^*$ and $v\in V_i$, $\strat{i}(\rho \cdot v)$  is a Dirac distribution, and it is called \emph{memoryless} if $\strat{i}(\rho  \cdot v) = \strat{i}(v)$, for every $\rho \in V^*$ and $v\in V_i$. 
Given two strategies $\strat{1} \in \Strategies{1}$, $\strat{2} \in \Strategies{2}$ and a starting state $v$,  the \emph{result} of the game  is a Markov chain, denoted by $\StochG^{\strat{1}, \strat{2}}_v$. 
%
As any Markov chain, $\StochG^{\strat{1}, \strat{2}}_v$ defines a
probability measure $\Prob{\strat{1}}{\strat{2}}_{\StochG,v}$ on the
Borel $\sigma$-algebra generated by the cylinders of $\Omega$.  If
$\mathcal{A}$ is a measurable set in such Borel
$\sigma$-algebra,
$\Prob{\strat{1}}{\strat{2}}_{\StochG,v}(\mathcal{A})$ is
the probability that strategies $\strat{1}$ and $\strat{2}$ generate a
play belonging to $\mathcal{A}$ from state $v$.
It would normally be convenient to use {\LTL} notation to define events. For instance,
$\Diamond V' =  \{ \rho = \rho_0,\rho_1,\dots \in \Omega \mid  \exists i : \rho_i \in V' \}$ defines the event in which some state in $V'$ is reached.
The outcome of the game, denoted by $\out_v(\strat{1}, \strat{2})$ is the set of possible paths
of $\StochG^{\strat{1}, \strat{2}}_{v}$ starting at vertex $v$ (i.e., the possible plays when strategies $\strat{1}$ and $\strat{2}$ are used).  When the initial state $v$ is fixed,  we write
$\out(\strat{1}, \strat{2})$ instead of $\out_{v}(\strat{1}, \strat{2})$. 

A \emph{Boolean objective} for $\StochG$ is a  set $\Phi \subseteq \Omega$.
A play $\rho$ is \emph{winning} for Player $1$ at vertex $v$ if $\rho \in \Phi$, otherwise 
it is winning for Player $2$ 
(i.e., we consider \emph{zero-sum} games).   A strategy $\strat{1}$ is a \emph{sure winning strategy} for Player $1$  from vertex $v$ if, for every strategy $\strat{2}$ for Player $2$, $\out_v(\strat{1}, \strat{2}) \subseteq \Phi$. $\strat{1}$ is said to be \emph{almost-sure winning} if for every strategy $\strat{2}$ for Player $2$, we have $\Prob{\strat{1}}{\strat{2}}_{\StochG,v}(\Phi)=1$.
	Sure  and almost-sure winning strategies for Player $2$ are defined in a similar way.
Reachability games are games with Boolean objectives of the style: 
$\Diamond V'$, for some set $V' \subseteq V$. A standard result is that, if a reachability game has a sure winning strategy, then
it has a pure memoryless sure winning strategy \cite{ChatterjeeH12}. 

A \emph{quantitative objective} is a measurable  function $f: \Omega \rightarrow \mathbb{R}$. Given a measurable function we define $\Expect{\strat{1}}{\strat{2}}_{\StochG,v}[f]$ as the expectation of
function $f$ under probability $\Prob{\strat{1}}{\strat{2}}_{\StochG,v}$. The goal of Player $1$ is to maximize the expected value of $f$, whereas the goal of 
Player $2$ is to minimize it.
Usually, quantitative objective functions are defined via a
\emph{reward function} $r:V \rightarrow \mathbb{R}$.
The value of the game for Player $1$ for strategy $\strat{1}$ at vertex $v$, 
denoted $\val_1(\strat{1})(v)$, is defined as: $\val_1(\strat{1})(v) = \inf_{\strat{2} \in \Strategies{2}} \Expect{\strat{1}}{\strat{2}}_{\StochG,v}[f]$.
Furthermore, the \emph{value of the game} for Player $1$ from vertex $v$ is defined as: $\sup_{\strat{1} \in \Strategies{1}} \val_1(\strat{1})(v)$.
Analogously, the value of the game for a Player $2$ strategy $\strat{2}$ and the value of the game 
for Player $2$ are defined as $\val_2(\strat{2})(v) = \sup_{\strat{1} \in \Strategies{1}}  \Expect{\strat{1}}{\strat{2}}_{\StochG,v} [f]$ 
and $\inf_{\strat{2} \in \Strategies{2}} \val_2(\strat{2})(v)$, respectively. We say that a game is determined if both values are equal, that is,
$\sup_{\strat{1} \in \Strategies{1}} \val_1(\strat{1})(v) = \inf_{\strat{2} \in \Strategies{2}} \val_2(\strat{2})(v)$, for every vertex $v$.

\section{Probabilistic Masking Simulation} \label{sec:probMaskDist}


We start this section by defining a probabilistic extension of the strong masking simulation introduced in
\cite{CastroDDP18b}.  Roughly speaking,  this is a variation of probabilistic bisimulation that takes into account the occurrence of faults (named masking simulation), and captures masking behavior.
This relation serves as a starting point for defining our masking games. We prove that in the Boolean case, our games allows us to decide masking simulation.  Since these games are infinite
we provide a finite symbolic characterization of them.  In Section \ref{sec:probAlmostSure}, we extend these games with quantitative objectives,  which allows us to quantify the level of fault-tolerance offered by an implementation. 

\paragraph*{The relation.}%
In simple terms, a probabilistic masking simulation is a relation
between PTSs that extends probabilistic bisimulation~\cite{journals/iandc/LarsenS91,Segala95}
in order to account for fault masking.  One of the PTSs
acts as the nominal model (or specification),  i.e., it describes the  behavior of the system when
no faults are considered, and the other one represents a possible fault-tolerant   implementation of the specification,  in which the occurrence of faults
are taken into account via a fault tolerance mechanism acting upon them.  

Probabilistic masking simulation allows one to analyze
whether the implementation is able to mask the faults while
preserving the behavior of the specification.  More specifically, for
non-faulty transitions, the relation behaves as probabilistic
bisimulation, which is captured by means of couplings and relations
respecting these couplings. 
The novel part is given by the occurrence of faults: if the
implementation performs a fault, the nominal model matches it by an idle step (this represents internal fault masking mechanisms).

In the following,  given  a set
of actions $\Sigma$, and a (finite) set of \emph{fault labels} $\faults$, with $\faults\cap\Sigma=\emptyset$, we define $\SigmaF
= \Sigma \cup \faults$.  Intuitively, the elements of $\faults$
indicate the occurrence of a fault in a faulty implementation.

\begin{definition} \label{def:maskingRel}
  Let $A =( S, \Sigma, {\rightarrow}, s_0 )$ and
  $A' =( S', \SigmaF, {\rightarrow'}, s_0' )$ be two PTSs representing
  the nominal and the implementation model, respectively.
  $A'$ is \emph{(strong) probabilistic masking fault-tolerant} with
  respect to $A$ iff there exists a relation $\M \subseteq S \times S'$
  such that:
  \begin{inparaenum}[(a)]
  \item%
    $s_0 \M s'_0$, and
  \item%
    for all $s \in S, s' \in S'$ with $s \M s'$ and all $e \in \Sigma$
    and $F \in \faults$ the following holds:
  \end{inparaenum}%
 \begin{enumerate}[(1)]
  \item%
    if $s \xrightarrow{e} \mu$, then $s' \xrightarrowprime{e} \mu'$ and
    $\mu \MaskCoup \mu'$ for some $\mu'$;
  \item%
    if $s' \xrightarrowprime{e} \mu'$, then $s \xrightarrow{e} \mu$ and
    $\mu \MaskCoup \mu'$ for some $\mu$;
  \item%
    if $s' \xrightarrowprime{F} \mu'$, then $\Dirac_s \MaskCoup \mu'$.
  \end{enumerate}
  %
  If such a relation exists we say that $A'$ is a \emph{(strong)
    probabilistic masking fault-tolerant implementation} of $A$, 
    denoted $A \Masking A'$.
\end{definition}

Note that the relation can be encoded in terms of traditional probabilistic
bisimulation as follows: saturate PTSs $A$ and $A'$ by adding
self-loops $s\xrightarrow{F}\Dirac_s$ and
$s'\xrightarrowprime{F}\Dirac_{s'}$, respectively, for every $s\in S$,
$s'\in S'$ and $F\in\faults$.  It follows from the definitions that these two new PTSs are
probabilistic bisimilar iff $A \Masking A'$.   As a consequence, checking $A \Masking A'$ is decidable in
polynomial time.  

%
%
\begin{example}\label{example:memory}
Consider a memory cell storing one bit of information that periodically refreshes its value.  The memory supports both write and read operations,  and when it refreshes, it
performs a read operation and overwrites the memory with the read value.
%
This behaviour is captured by the nominal model of
Fig.~\ref{fig:exam1MemCell:nom} using  {\PRISM} 
notation~\cite{DBLP:conf/cav/KwiatkowskaNP11}.  
In this model,  $\codett{ri}$
and $\codett{wi}$ (for $\codett{i}=\codett{0,1}$) represent the actions of reading and
writing  value $\codett{i}$.  The bit stored in the memory is saved in
variable~\codett{b}.  Action \codett{tick} marks that one time unit has passed
and, with probability \codett{p}, it
enables the refresh action (\codett{rfsh}).  Variable \codett{m}
indicates whether the system is in write/read mode, or producing a
refresh.
\begin{figure}
\begin{minipage}[t]{.4\textwidth}
{\fontsize{9}{9}\selectfont\ttfamily
\begin{tabbing}
x\=xxxxxxx\=xxxxxxxxxxxxx\=x\=xxx\= \kill    
module NOMINAL\\[1ex]
\>b : [0..1] init 0;\\
\>m : [0..1] init 0; \>\>// 0 = normal,\\
\>                   \>\>// 1 = refreshing\\[1ex]
\>[w0]   \>(m=0)          \>\>-> \>(b'= 0);\\
\>[w1]   \>(m=0)          \>\>-> \>(b'= 1);\\
\>[r0]   \>(m=0) \& (b=0) \>\>-> \>true;\\
\>[r1]   \>(m=0) \& (b=1) \>\>-> \>true;\\
\>[tick] \>(m=0)          \>\>-> \>p: (m'= 1) +\\
\>       \>               \>\>   \>(1-p): true;\\
\>[rfsh] \>(m=1)          \>\>-> \>(m'= 0);\\[1ex]
endmodule\\[3.7em]
\end{tabbing}}
\caption{Memory cell: nominal model}\label{fig:exam1MemCell:nom}
\end{minipage}
\hfill
\begin{minipage}[t]{.55\textwidth}
{\fontsize{9}{9}\selectfont\ttfamily
\begin{tabbing}
x\=xxxxxxxx\=xxxxxxxxxxxxx\=xxx\=xxx\= \kill    
module FAULTY\\[1ex]
\>v : [0..3] init 0;\\
\>s : [0..2] init 0; \>\>// 0 = normal, 1 = faulty,\\
\>                   \>\>// 2 = refreshing\\
\>\textcolor{red}{f : [0..1] init 0;} \>\>\textcolor{red}{// fault limiting artifact}\\[1ex]
\>[w0]    \>(s!=2)           \>\>-> \>(v'= 0) \& (s'= 0);\\
\>[w1]    \>(s!=2)           \>\>-> \>(v'= 3) \& (s'= 0);\\
\>[r0]    \>(s!=2) \& (v<=1) \>\>-> \>true;\\
\>[r1]    \>(s!=2) \& (v>=2) \>\>-> \>true;\\
\>[tick]  \>(s!=2)           \>\>-> \>p: (s'= 2) + q: (s'= 1) \\
\>        \>                 \>\>   \>+ (1-p-q): true;\\
\>[rfsh]  \>(s=2)            \>\>-> \>(s'=0)\\
\>        \>                 \>\>   \>\& (v'= (v<=1) ? 0 : 3);\\
\>[fault] \>(s=1) \textcolor{red}{\& (f<1)}   \>\>-> \>(v'= (v<3) ? (v+1) : 2) \\
\>        \>                 \>\>   \>\& (s'= 0) \textcolor{red}{\& (f'= f+1)};\\
\>[fault] \>(s=1) \textcolor{red}{\& (f<1)}    \>\>-> \>(v'= (v>0) ? (v-1) : 1) \\
\>        \>                 \>\>   \>\& (s'= 0) \textcolor{red}{\& (f'= f+1)};\\[1ex]
endmodule\\[-2em]
\end{tabbing}}
\caption{Memory cell: fault-tolerant implementation.}\label{fig:exam1MemCell:ft}
\end{minipage}
\end{figure}

%
%
A potential fault in this scenario occurs when a cell unexpectedly changes its value. 
In practice, the occurrence of such an error 
has a certain probability. A typical technique to deal with this situation is \emph{redundancy}, 
e.g., using three memory bits instead of one. Then, writing operations are performed simultaneously 
on the
three bits while reading returns the value read by majority (or \emph{voting}).
%
Fig.~\ref{fig:exam1MemCell:ft} shows this implementation with
the occurrence of the fault implicitly modeled (ignore, for the time
being, the red part).  Variable \codett{v} counts the votes for the value
1.
%
In addition to enabling the refresh action, a \codett{tick} may also
enable the occurrence of a fault with probability \codett{q}, with
$\codett{p}+\codett{q}\leq 1$.
Variable \codett{s} indicates whether the system is in normal mode
($\codett{s}=0$), in a state where a fault may occur ($\codett{s}=1$),
or producing a refresh ($\codett{s}=2$).
%
The red coloured text in Fig.~\ref{fig:exam1MemCell:ft} is an artifact to limit the number
of faults to $1$.  Under this condition,
relation
$\M = $ $\{{\langle(b,m),(v,s,f)\rangle} \mid {{2b\leq v\leq 2b{+}1} \wedge (m=1 \Leftrightarrow s=2)}\}$
is a probabilistic masking simulation ($b$, $m$, $v$, $s$, and
$f$ represent the values of variables \codett{b}, \codett{m},
\codett{v}, \codett{s}, and \codett{f}, respectively.)
It should be evident that,  when the red coloured text is removed,  
\codett{FAULTY} is not a masking fault-tolerant
implementation of \codett{NOMINAL}. 
\end{example}

\paragraph*{A characterization in terms of  stochastic games.}


We define a stochastic masking simulation game for any given nominal model 
$A = ( S, \Sigma, {\rightarrow}, s_0 )$ and  implementation model $A' = ( S', \Sigma_{\mathcal{F}}, {\rightarrow'}, s'_0 )$. 
The game is similar to a bisimulation game \cite{Stirling98}, and it is played by two players, 
named for convenience the Refuter ($\Refuter$) and the Verifier ($\Verifier$). The Verifier 
wants to prove that $s \in S$ and $s' \in S'$ are probabilistic masking similar, 
and the Refuter intends to disprove that.
The game starts from the pair of states $(s,s')$ and the following steps are repeated:
\begin{enumerate}
\item[1)]
  $\Refuter$ chooses either a transition $s \xrightarrow{a} \mu$ from
  the nominal model or a transition $s' \xrightarrowprime{a} \mu'$ from the
  implementation;
\item[2a)]
  If $a \notin \faults$, $\Verifier$ chooses a transition matching
  action $a$ from the opposite model, i.e., a transition $s'
  \xrightarrowprime{a} \mu'$ if $\Refuter$'s choice was from the nominal model,
  or a transition $s \xrightarrow{a} \mu$ otherwise.  In addition,
  $\Verifier$ chooses a coupling $w$ for $(\mu, \mu')$;
\item[2b)]
  If $a \in \faults$, $\Verifier$ can only select the Dirac
  distribution $\Dirac_{s}$ and the only possible coupling $w$ for
  $(\Dirac_{s}, \mu')$;
\item[3)]
  The successor pair of states $(t, t')$ is chosen probabilistically
  according to $w$.
\end{enumerate}
If the play continues forever, then the Verifier wins; otherwise, the
Refuter wins. (Notice, in particular, that the Verifier loses if she cannot match a transition label, since choosing an arbitrary coupling is always possible.)  Step 2b is the only one that seems to differ from the
usual bisimulation game.  This is needed because of the asymmetry produced by the
transitions labeled with  faults. Intuitively,  if the Refuter
chooses to play a fault in the implementation, then the Verifier ought to
mask the fault,  thus she cannot freely move in the
nominal model.  Summing up,  the probabilistic step of a fault can only be
matched by a Dirac distribution on the corresponding state of the specification.

In the following we define the stochastic masking game graph that
formalizes this idea.
For this, define $\Sigma^i = \{ e^i \mid e \in \Sigma\}$
containing all elements of $\Sigma$ indexed with superscript $i$.

\begin{definition} \label{def:strongMaskingGameGraphi}
  Let $A =( S, \Sigma, {\rightarrow}, s_0 )$ and
  $A' = ( S', \Sigma_\faults , {\rightarrow'}, s'_0 )$ be two PTSs.
  The 2-player \emph{stochastic masking game graph}
  $\StochG_{A,A'} = (V^\StochG, E^\StochG, V^\StochG_\Refuter, V^\StochG_\Verifier, V^\StochG_\Probabilistic, \InitVertex, \delta^\StochG)$,
   is defined as follows:
%
  \begin{align*}
    V^\StochG = \
    & V^\StochG_\Refuter \cup V^\StochG_\Verifier \cup V^\StochG_\Probabilistic, \text{where: }\\
    V^\StochG_\Refuter = \
    & \{ (s, \mhyphen, s', \mhyphen, \mhyphen, \mhyphen, \Refuter) \mid
          s \in S \wedge s' \in S' \} \cup
      \{\ErrorSt\}\\
    V^\StochG_\Verifier = \
    & \{ (s, \sigma^1, s', \mu, \mhyphen, \mhyphen, \Verifier) \mid
         s \in S \wedge s' \in S'
         \wedge \sigma \in \Sigma
         \wedge s\xrightarrow{\sigma} \mu \} \cup {} \\
    & \{ (s, \sigma^2, s', \mhyphen, \mu', \mhyphen, \Verifier) \mid
         s \in S \wedge s' \in S'
         \wedge \sigma \in \SigmaF
         \wedge s'\xrightarrowprime{\sigma}\mu'\}\\
    V^\StochG_\Probabilistic = \
    & \{ (s, \mhyphen, s', \mu, \mu', w, \Probabilistic) \mid
         s \in S \wedge s' \in S' 
         \wedge w \in \couplings{\mu}{\mu'} \wedge {}\\
         & \phantom{\{ (s, \mhyphen, s', \mu, \mu', w, \Probabilistic) \mid {}}
         \exists \sigma{\in}\SigmaF :
         {({s\xrightarrow{\sigma} \mu} \vee (\sigma{\in}\faults \wedge \mu=\Dirac_s))}
         \wedge {s'\xrightarrowprime{\sigma}\mu'}\}\\
    \InitVertex = \
    & ( s_0, \mhyphen, s'_0, \mhyphen, \mhyphen, \mhyphen, \Refuter )
    \text{ \ (the Refuter starts playing)}\\
    & \hspace{-2.6em}
    \delta^\StochG : V^\StochG_\Probabilistic \rightarrow \Dist(V^\StochG_\Refuter),
      \text{ defined by } 
      \delta^\StochG((s, \mhyphen, s', \mu, \mu', w, \Probabilistic))((t, \mhyphen, t', \mhyphen, \mhyphen, \mhyphen, \Refuter)) = w(t,t')
      \text{,}
  \end{align*}
%
  where ``$\mhyphen$'' fills an unused place,
  and $E^\StochG$ is the minimal set satisfying the following rules:
%
  \begin{align*}
    s \xrightarrow{\sigma} \mu
    & \Rightarrow \tuple{(s, \mhyphen, s', \mhyphen, \mhyphen, \mhyphen, \Refuter), (s, \sigma^{1}, s', \mu, \mhyphen, \mhyphen, \Verifier)}\in E^\StochG
    \tag{1$_1$}\label{play:r:1}\\
    s' \xrightarrowprime{\sigma} \mu'
    & \Rightarrow \tuple{(s, \mhyphen, s', \mhyphen, \mhyphen, \mhyphen, \Refuter),(s, \sigma^{2}, s', \mhyphen, \mu', \mhyphen, \Verifier)}\in E^\StochG
    \tag{1$_2$}\label{play:r:2}\\
    {s' \xrightarrowprime{\sigma} \mu'} \wedge {w \in \couplings{\mu}{\mu'}}
    & \Rightarrow \tuple{(s, \sigma^1, s', \mu, \mhyphen, \mhyphen, \Verifier),(s, \mhyphen, s', \mu, \mu', w, \Probabilistic)}\in E^\StochG
    \tag{2a$_1$}\label{play:v:1}\\
    \!\!\!{\sigma \notin \faults} \wedge {s \xrightarrow{\sigma} \mu} \wedge {w \in \couplings{\mu}{\mu'}}
    & \Rightarrow \tuple{(s, \sigma^2, s', \mhyphen, \mu', \mhyphen, \Verifier), (s, \mhyphen, s', \mu, \mu', w, \Probabilistic)}\in E^\StochG
    \tag{2a$_2$}\label{play:v:2a}\\
    {F \in \faults} \wedge {w \in \couplings{\Dirac_s}{\mu'}}
    & \Rightarrow \tuple{(s, F^2, s', \mhyphen, \mu', \mhyphen, \Verifier), (s, \mhyphen, s', \Dirac_s, \mu', w, \Probabilistic)}\in E^\StochG
    \tag{2b}\label{play:v:2b}\\
    (s, \mhyphen, s', \mu, \mu', w, \Probabilistic) \in V^\StochG_\Probabilistic \wedge {} & 
    (t,t') \in \support{w} 
    \Rightarrow \tuple{(s, \mhyphen, s', \mu, \mu', w, \Probabilistic), (t, \mhyphen, t', \mhyphen, \mhyphen, \mhyphen, \Refuter)}\in E^\StochG
    \tag{3}\label{play:p}\\
    {v\in (V^\StochG_\Verifier{\cup}\{\ErrorSt\})} & {} \wedge (\nexists v' \neq\ErrorSt : {\tuple{v,v'}\in E^\StochG})
    \Rightarrow  \tuple{v,\ErrorSt}\in E^\StochG
    \tag{err}\label{play:err}
  \end{align*}
  %
\end{definition}

Some words about this definition are useful, it mainly follows the idea of the game previously described.   A
round of the game starts in the Refuter's state $\InitVertex$.  Notice
that, at this point, only the current states of the nominal and
implementation models are relevant (all other information is not yet defined in this round and hence marked with
``$\mhyphen$'').  Step 1 of the game is encoded in rules
(\ref{play:r:1}) and (\ref{play:r:2}), where the Refuter chooses a
transition, thus defining the action and distribution that need to be
matched,  this moves the game to a Verifier's state.  A Verifier's state
in $V^\StochG_\Verifier$ is a tuple containing which action
and distribution need to be matched, and which model the Refuter has
played.  Step 2a of the game is given by rules (\ref{play:v:1}) and
(\ref{play:v:2a}) in which the Verifier chooses a matching move from
the opposite model (hence defining the other distribution) and an
appropriate coupling,  moving to a probabilistic state.  Step 2b
of the game is encoded in rule (\ref{play:v:2b}).  Here the Verifier
has no choice since she is obliged to choose the Dirac distribution
$\Dirac_s$ and the only available coupling in
$\couplings{\Dirac_s}{\mu'}$.  A probabilistic state in
$V^\StochG_\Probabilistic$ contains the information needed to probabilistically
resolve the next step through function $\delta^\StochG$ (rule
(\ref{play:p})).
Finally,  rule (\ref{play:err}) states that, if a player has no move,  then she reaches an error state ($\ErrorSt$).  Note that this can only happen in a Verifier's state or in
$\ErrorSt$.

The notion of probabilistic masking simulation can be captured by the
corresponding stochastic masking game with the appropriate Boolean
objective.

\begin{theorem} \label{thm:wingameStrat}
  Let $A= ( S, \Sigma, {\rightarrow}, s_0 )$ and $A'=( S', \SigmaF, {\rightarrow'}, s_0' )$ be two PTSs.
  Then, $A \Masking A'$ iff the Verifier has a sure (or almost-sure) winning strategy for the stochastic masking game graph 
  $\mathcal{G}_{A,A'}$ with the Boolean objective $\neg \Diamond \ErrorSt$.
\end{theorem}
Note that this theorem holds for  both sure and almost-sure strategies of the Verifier,  
this follows from the fact that for stochastic reachability objectives the two kinds of strategies are equivalent.

\definecolor{shadecolor}{RGB}{190,220,255}
\newcommand{\bkgrblue}[1]{\makebox[0pt][l]{\hspace{-.4ex}\raisebox{0ex}[0ex][0ex]{\color{shadecolor}\rule[-.8ex]{#1}{2.9ex}}}}

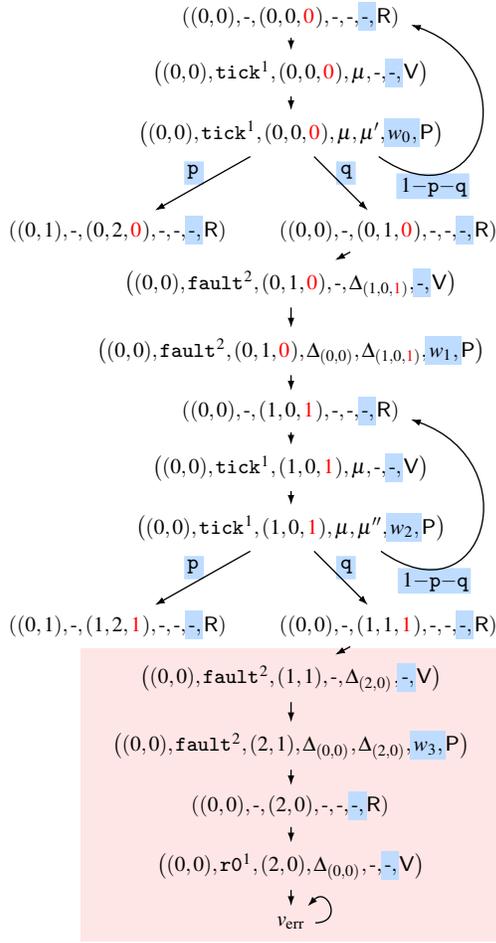
\begin{figure}[t]
\begin{minipage}[t]{.47\textwidth}
\begin{center}
  \scalebox{1}{
    \begin{tikzpicture}[on grid,auto,align at top,,
        hvh path/.style={to path={-- ++(#1,0) |- (\tikztotarget)},rounded corners}]

      \node (ar) {$\left((0,0),\mhyphen,(0,0,\textcolor{red}{0}),\mhyphen,\mhyphen,\bkgrblue{1.8ex}\mhyphen,\Refuter\right)$};
      \node (av) [below=0.75 of ar]  {$\left((0,0),\codett{tick}^1,(0,0,\textcolor{red}{0}),\mu,\mhyphen,\bkgrblue{1.8ex}\mhyphen,\Verifier\right)$};
      \node (ap) [below=0.8 of av]  {$\left((0,0),\codett{tick}^1,(0,0,\textcolor{red}{0}),\mu,\mu',\bkgrblue{3.7ex}w_0,\Probabilistic\right)$};
      \node (br) [below right=1.3 and 1.3 of ap] {$\left((0,0),\mhyphen,(0,1,\textcolor{red}{0}),\mhyphen,\mhyphen,\bkgrblue{1.8ex}\mhyphen,\Refuter\right)$};
      \node (brign) [below left=1.3 and 2.3 of ap] {$\left((0,1),\mhyphen,(0,2,\textcolor{red}{0}),\mhyphen,\mhyphen,\bkgrblue{1.8ex}\mhyphen,\Refuter\right)$};
      \node (bv) [below left=.7 and 1.3 of br] {$\left((0,0),\codett{fault}^2,(0,1,\textcolor{red}{0}),\mhyphen,\Dirac_{(1,0,\textcolor{red}{1})},\bkgrblue{1.8ex}\mhyphen,\Verifier\right)$};
      \node (bp) [below =.9 of bv] {$\left((0,0),\codett{fault}^2,(0,1,\textcolor{red}{0}),\Dirac_{(0,0)},\Dirac_{(1,0,\textcolor{red}{1})},\bkgrblue{3.7ex}w_1,\Probabilistic\right)$};
      \node (cr) [below =.8 of bp] {$\left((0,0),\mhyphen,(1,0,\textcolor{red}{1}),\mhyphen,\mhyphen,\bkgrblue{1.8ex}\mhyphen,\Refuter\right)$};
      \node (cv) [below=.75 of cr]  {$\left((0,0),\codett{tick}^1,(1,0,\textcolor{red}{1}),\mu,\mhyphen,\bkgrblue{1.8ex}\mhyphen,\Verifier\right)$};
      \node (cp) [below=.8 of cv]  {$\left((0,0),\codett{tick}^1,(1,0,\textcolor{red}{1}),\mu,\mu'',\bkgrblue{3.7ex}w_2,\Probabilistic\right)$};
      \node (dr) [below right=1.3 and 1.3 of cp] {$\left((0,0),\mhyphen,(1,1,\textcolor{red}{1}),\mhyphen,\mhyphen,\bkgrblue{1.8ex}\mhyphen,\Refuter\right)$};
      \node (drign) [below left=1.3 and 2.3 of cp] {$\left((0,1),\mhyphen,(1,2,\textcolor{red}{1}),\mhyphen,\mhyphen,\bkgrblue{1.8ex}\mhyphen,\Refuter\right)$};
      \node (dv) [below left=.7 and 1.3 of dr] {$\left((0,0),\codett{fault}^2,(1,1),\mhyphen,\Dirac_{(2,0)},\bkgrblue{1.8ex}\mhyphen,\Verifier\right)$};
      \node (dp) [below =.9 of dv] {$\left((0,0),\codett{fault}^2,(2,1),\Dirac_{(0,0)},\Dirac_{(2,0)},\bkgrblue{3.7ex}w_3,\Probabilistic\right)$};
      \node (er) [below =.8 of dp] {$\left((0,0),\mhyphen,(2,0),\mhyphen,\mhyphen,\bkgrblue{1.8ex}\mhyphen,\Refuter\right)$};
      \node (ev) [below =.8 of er] {$\left((0,0),\codett{r0}^1,(2,0),\Dirac_{(0,0)},\mhyphen,\bkgrblue{1.8ex}\mhyphen,\Verifier\right)$};
      \node (err) [below =.75 of ev] {$\ErrorSt$};

      
      
      \path[-{Latex[length=1.35mm,width=0.9mm]},line width=0.5pt]
      (ar) edge[] (av)
      (av) edge[] (ap)
      (ap) edge[] node[above,pos=0.6] {$\codett{q}$} (br)
      (ap) edge[] node[above,pos=0.6] {$\codett{p}$} (brign)
      ($ (ap.south east) - (10mm,0) $) edge[looseness=2.4,out=325,in=340] node[below,pos=0.2,yshift=1mm] {$1{-}\codett{p}{-}\codett{q}$} ($ (ar.east) - (0,1mm) $)
      (br) edge[] (bv)
      (bv) edge[] (bp)
      (bp) edge[] (cr)
      (cr) edge[] (cv)
      (cv) edge[] (cp)
      (cp) edge[] node[above,pos=0.6] {$\codett{q}$} (dr)
      (cp) edge[] node[above,pos=0.6] {$\codett{p}$} (drign)
      ($ (cp.south east) - (10mm,0) $) edge[looseness=2.4,out=325,in=340] node[below,pos=0.2,yshift=1mm] {$1{-}\codett{p}{-}\codett{q}$} ($ (cr.east) - (0,1mm) $)
      (dr) edge[] (dv)
      (dv) edge[] (dp)
      (dp) edge[] (er)
      (er) edge[] (ev)
      (ev) edge[] (err)
      (err) edge[loop,out=0,in=45,looseness=4] (err)
      ;

      \scoped[on background layer]{
        \node (fznw) [above left=.4 and 2.8 of dv] {};
        \node (fzse) [below right=3.9 and 5.6 of fznw] {};

        \path[fill=red!10]
        (fznw) rectangle (fzse)
        ;
      }
      \scoped[on background layer]{
        \node (fzpnw) [below left=.37 and 1.42 of ap] {};
        \node (fzpse) [below right=.31 and .26 of fzpnw] {};
        \node (fzqnw) [below right=.37 and .6 of ap] {};
        \node (fzqse) [below right=.31 and .26 of fzqnw] {};
        \node (fzrnw) [below right=.53 and 1.4 of ap] {};
        \node (fzrse) [below right=.34 and .97 of fzrnw] {};

        \path[fill=shadecolor]
        (fzpnw) rectangle (fzpse)
        (fzqnw) rectangle (fzqse)
        (fzrnw) rectangle (fzrse)
        ;
      }
      \scoped[on background layer]{
        \node (fzpnw) [below left=.37 and 1.42 of cp] {};
        \node (fzpse) [below right=.31 and .26 of fzpnw] {};
        \node (fzqnw) [below right=.37 and .6 of cp] {};
        \node (fzqse) [below right=.31 and .26 of fzqnw] {};
        \node (fzrnw) [below right=.53 and 1.42 of cp] {};
        \node (fzrse) [below right=.34 and .97 of fzrnw] {};

        \path[fill=shadecolor]
        (fzpnw) rectangle (fzpse)
        (fzqnw) rectangle (fzqse)
        (fzrnw) rectangle (fzrse)
        ;
      }
    \end{tikzpicture}
  }
  \end{center}
  \end{minipage}
  \hspace{1.9em}
  \begin{minipage}[t]{.472\textwidth}
  \small\vspace{-1.4em}
  \begin{align*}
    \mu   &= \codett{p}\cdot(0,1)+(1{-}\codett{p})\cdot(0,0)\\[1ex]
    \mu'  &= \begin{cases}\codett{p}\cdot(0,2,\textcolor{red}{0})+\codett{q}\cdot(0,1,\textcolor{red}{0})+{}\\(1{-}\codett{p}{-}\codett{q})\cdot(0,0,\textcolor{red}{0})\end{cases}\\[1ex]
    \mu'' &= \begin{cases}\codett{p}\cdot(1,2,\textcolor{red}{1})+\codett{q}\cdot(1,1,\textcolor{red}{1})+{}\\(1{-}\codett{p}{-}\codett{q})\cdot(1,0,\textcolor{red}{1})\end{cases}\\[1ex]
  w_0   &= \begin{cases}\codett{p}\cdot((0,1),(0,2,\textcolor{red}{0}))+\codett{q}\cdot((0,0),(0,1,\textcolor{red}{0}))+{}\\(1{-}\codett{p}{-}\codett{q})\cdot((0,0),(0,0,\textcolor{red}{0}))\end{cases}\\[1ex]
    w_1   &= \Dirac_{((0,0),(1,0,\textcolor{red}{1}))}\\[1ex]
    w_2   &= \begin{cases}\codett{p}\cdot((0,1),(1,2,\textcolor{red}{1}))+\codett{q}\cdot((0,0),(1,1,\textcolor{red}{1}))+{}\\(1{-}\codett{p}{-}\codett{q})\cdot((0,0),(1,0,\textcolor{red}{1}))\end{cases}\\[1ex]
    w_3   &= \Dirac_{((0,0),(2,0))}
  \end{align*}
  \end{minipage}
  
  \caption{A fragment of a masking game graph}\label{fig:masking:game:graph}
\end{figure}

\begin{example}\label{ex:running}
  Consider the graph in Fig.~\ref{fig:masking:game:graph} (ignoring the blue shading for now). 
  It represents a fragment of the masking game graph between
  \codett{NOMINAL} and \codett{FAULTY} of Example~\ref{example:memory}.
  The vertices represent the variable values in the following order:
  $((\codett{b},\codett{m}),\_,(\codett{v},\codett{s},\textcolor{red}{\codett{f}}),\_,\_,\_,\_)$.
  First,  consider the graph disregarding the red highlighted numbers.  For example,
  $\left((0,0),\mhyphen,(0,0,\textcolor{red}{0}),\mhyphen,\mhyphen,\mhyphen,\Refuter\right)$
  should be read as
  $\left((0,0),\mhyphen,(0,0),\mhyphen,\mhyphen,\mhyphen,\Refuter\right)$.
  In this case we obtain the masking game graph when the red part in
  \codett{FAULTY} is removed.
  Notice that, in the majority of the vertices, many outgoing edges are
  omitted.  In particular, the Verifier vertex
  $\left((0,0),\codett{tick}^1,(0,0),\mu,\mhyphen,\mhyphen,\Verifier\right)$
  has infinitely many
  outgoing edges leading to probabilistic vertices of the  form
  $\left((0,0),\codett{tick}^1,(0,0),\mu,\mu',w,\Probabilistic\right)$,
  where $w$ is a coupling for $(\mu,\mu')$.
  In the graph, we have chosen to distinguish coupling $w_0$ which is
  optimal for the Verifier (similarly later for $w_2$).
  We highlighted the path leading to error state $\ErrorSt$.
  Notice that this occurs as a consequence of the Refuter choosing to
  do a second \codett{fault} in vertex
  $\left((0,0),\mhyphen,(1,1),\mhyphen,\mhyphen,\mhyphen,\Refuter\right)$
  steering the game to the red shadowed part of the graph.  Later, the
  Refuter chooses to read 0 in the \codett{NOMINAL} model (at vertex
  $\left((0,0),\mhyphen,(2,0),\mhyphen,\mhyphen,\mhyphen,\Refuter\right)$)
  which the Verifier cannot match.

  Now,  consider the masking game graph between \codett{NOMINAL} and
  fault-limited \codett{FAULTY} model (i.e.,  take now into account the
  red part).  This graph includes the red values corresponding to
  variable $\textcolor{red}{\codett{f}}$.  Notice that here, the
  Refuter cannot 
  produce a \codett{fault} transition from
  vertex
  $\left((0,0),\mhyphen,(1,1,\textcolor{red}{\codett{1}}),\mhyphen,\mhyphen,\mhyphen,\Refuter\right)$.
  Thus,  in this case, the Verifier manages to avoid reaching the error state
  $\ErrorSt$.
\end{example}
\paragraph*{A symbolic game graph.}

The graph for a stochastic masking game could be infinite
since each probabilistic node includes a coupling between the two
contending distributions, and there can be uncountably many of them.
In the following, we introduce a finite description of stochastic
masking games through a symbolic representation that omits explicit
reference to couplings.
%
The definition of the symbolic game graph is twofold.
The first part captures the non-stochastic behaviour of the
game by removing the stochastic choice ($\delta^\StochG$) of the
graph as well as the couplings on the vertices.  The
second part appends an equation system to each probabilistic vertex
whose solution space is the polytope defined by the set of all
couplings for the contending distributions.

\begin{definition} \label{def:symbolicGameGraph}
  Let $A = ( S, \Sigma, {\rightarrow}, s_0 )$
  and $A' = ( S', \Sigma_\faults, {\rightarrow'}, s'_0 )$
  be two PTSs.
  The \emph{symbolic game graph} for the stochastic masking game
  $\mathcal{G}_{A,A'}$ is defined by 
  $\SymbG_{A,A'} = ( V^{\SymbG}, E^{\SymbG}, V^{\SymbG}_\Refuter, V^{\SymbG}_\Verifier, V^{\SymbG}_\Probabilistic, \InitVertexSG )$,
  where:
  \begin{align*}
    V^\SymbG = \
    & V^\SymbG_\Refuter \cup V^\SymbG_\Verifier \cup V^\SymbG_\Probabilistic, \text{ where: }\\
    V^\SymbG_\Refuter = \
    & \{ (s, \mhyphen, s', \mhyphen, \mhyphen, \Refuter) \mid
          s \in S \wedge s' \in S' \} \cup
      \{\ErrorSt\}\\
    V^\SymbG_\Verifier = \
    & \{ (s, \sigma^1, s', \mu, \mhyphen, \Verifier) \mid
         s \in S \wedge s' \in S'
         \wedge \sigma \in \Sigma
         \wedge s \xrightarrow{\sigma} \mu \} \cup {} \\
    & \{ (s, \sigma^2, s', \mhyphen, \mu', \Verifier) \mid
         s \in S \wedge s' \in S'
         \wedge \sigma \in \SigmaF
          \wedge s' \xrightarrowprime{\sigma} \mu' \} \\
    V^\SymbG_\Probabilistic = \
    & \{ (s, \mhyphen, s', \mu, \mu', \Probabilistic) \mid
         s \in S \wedge s' \in S' \wedge  
         \exists \sigma{\in}\SigmaF : (s\xrightarrow{\sigma}\mu \wedge (\sigma{\in}\faults \vee \mu=\Dirac_s)) \wedge s'\xrightarrowprime{\sigma}\mu'\}\\
    \InitVertexSG = \
    & ( s_0, \mhyphen, s'_0, \mhyphen, \mhyphen, \Refuter ),
  \end{align*}
%
  and $E^\SymbG$ is the minimal set satisfying the following rules:
%
  \begin{align*}
    s \xrightarrow{\sigma} \mu
    & \Rightarrow \tuple{(s, \mhyphen, s', \mhyphen, \mhyphen, \Refuter), (s, \sigma^{1}, s', \mu, \mhyphen, \Verifier)}\in E^\SymbG \\
    s' \xrightarrowprime{\sigma} \mu'
    & \Rightarrow \tuple{(s, \mhyphen, s', \mhyphen, \mhyphen, \Refuter),(s, \sigma^{2}, s', \mhyphen, \mu', \Verifier)}\in E^\SymbG \\
    {s' \xrightarrowprime{\sigma} \mu'}
    & \Rightarrow \tuple{(s, \sigma^1, s', \mu, \mhyphen, \Verifier),(s, \mhyphen, s', \mu, \mu', \Probabilistic)}\in E^\SymbG \\
    {\sigma \notin \faults} \wedge {s \xrightarrow{\sigma} \mu}
    & \Rightarrow \tuple{(s, \sigma^2, s', \mhyphen, \mu', \Verifier), (s, \mhyphen, s', \mu, \mu', \Probabilistic)}\in E^\SymbG\\
    {F \in \faults}
    & \Rightarrow \tuple{(s, F^2, s', \mhyphen, \mu', \Verifier), (s, \mhyphen, s', \Dirac_s, \mu', \Probabilistic)}\in E^\SymbG \\
     (s, \mhyphen, s', \mu, \mu', \Probabilistic) \in V^\SymbG_\Probabilistic \wedge {} & 
    {t\in \support{\mu}} \wedge {t'\in\support{\mu'}}
      \Rightarrow \tuple{(s, \mhyphen, s', \mu, \mu', \Probabilistic), (t, \mhyphen, t', \mhyphen, \mhyphen,  \Refuter)}\in E^\SymbG \\
    {v\in (V^\SymbG_\Verifier{\cup}\{\ErrorSt\})} \wedge {} & {(\nexists {v'\neq\ErrorSt} : {\tuple{v,v'}\in E^\SymbG})}
    \Rightarrow  \tuple{v,\ErrorSt}\in E^\SymbG
  \end{align*}
  In addition, for each
  $v=(s, \mhyphen, s', \mu, \mu', \Probabilistic) \in V^\SymbG_\Probabilistic$,
  consider the set of variables
  $X(v)=\{x_{s_i,s_j} \mid s_i \in \text{Supp}(\mu) \wedge s_j \in \text{Supp}(\mu')\}$,
  and the system of equations
  %
  \begin{align*}
    \Eq(v) = {}
    & \textstyle
    \big\{ \sum_{s_j \in \support{\mu'}} x_{s_k,s_j}=\mu(s_k) \mid s_k \in \support{\mu} \big\} \cup {} \\
    & \textstyle
    \big\{ \sum_{s_k \in \support{\mu}} x_{s_k,s_j}=\mu'(s_j) \mid s_j \in \support{\mu'} \big\} \cup {} \\
    & \textstyle
    \big\{ x_{s_k,s_j} \geq 0 \mid s_k \in \support{\mu} \wedge s_j \in \support{\mu'} \big\}
  \end{align*}
\end{definition} 
Notice that $\{\bar{x}_{s_k,s_j}\}_{s_k,s_j}$ is a solution of
$\Eq(v)$ if and only if there is a coupling
$w\in\couplings{\mu}{\mu'}$ such that $w(s_k,s_j)=\bar{x}_{s_k,s_j}$
for all $s_k \in \support{\mu}$ and $s_j \in \support{\mu'}$.

Furthermore, given a set of game vertices
$V' \subseteq V^\SymbG_\Refuter$,
we define $\Eq(v, V')$ by extending $\Eq(v)$ with an equation limiting  the couplings in such a way that vertices in $V'$ are \emph{not} reached.
Formally,
$\Eq(v, V') = \Eq(v) \cup \big\{\sum_{(s, \mhyphen, s', \mhyphen, \mhyphen, \Refuter) \in V'} x_{s,s'} = 0\big\}$.
%
By properly defining a family of sets $V'$, we will show that
the stochastic masking game can be solved in polynomial time
through the symbolic game graph.

\begin{example}
  The fragment of the symbolic game graph of Example~\ref{example:memory}
  in Fig.~\ref{fig:masking:game:graph} is the same as depicted there
  only that all blue shaded components should be removed. (We also have
  the two variants here: one with the red values and the other
  one without them.)
  In the symbolic graph, vertex
  $v=\left((0,0),\codett{tick},(0,0,\textcolor{red}{0}),\mu,\mhyphen,\Verifier\right)$,  for example, has
  only one successor, in contraposition to vertex
  $\left((0,0),\codett{tick},(0,0,\textcolor{red}{0}),\mu,\mhyphen,\mhyphen,\Verifier\right)$
  that has uncountably many in the original game graph.
  Instead, $v$ has associated the set $\Eq(v)$ containing the following
  equations
%
  \begin{align*}
    & x_{(0,1),(0,2,\textcolor{red}{0})} + x_{(0,1),(0,1,\textcolor{red}{0})} + x_{(0,1),(0,0,\textcolor{red}{0})} = \codett{p} &
    & x_{(0,1),(0,2,\textcolor{red}{0})} + x_{(0,0),(0,2,\textcolor{red}{0})} = \codett{p} \\
    & x_{(0,0),(0,2,\textcolor{red}{0})} + x_{(0,0),(0,1,\textcolor{red}{0})} + x_{(0,0),(0,0,\textcolor{red}{0})} = 1-\codett{p} &
    & x_{(0,1),(0,0,\textcolor{red}{0})} + x_{(0,0),(0,0,\textcolor{red}{0})} = 1-\codett{p}-\codett{q} \\
    & x_{(0,1),(0,2,\textcolor{red}{0})} \geq 0 \qquad\quad
      x_{(0,1),(0,1,\textcolor{red}{0})} \geq 0 \qquad\quad
      x_{(0,1),(0,0,\textcolor{red}{0})} \geq 0 &
    & x_{(0,1),(0,1,\textcolor{red}{0})} + x_{(0,0),(0,1,\textcolor{red}{0})} = \codett{q} \\
    & x_{(0,0),(0,2,\textcolor{red}{0})} \geq 0 \qquad\quad
      x_{(0,0),(0,1,\textcolor{red}{0})} \geq 0 \qquad\quad
      x_{(0,0),(0,0,\textcolor{red}{0})} \geq 0
  \end{align*}
  In particular, notice that, if $w_0$ is as defined in Example~\ref{ex:running}, $\bar{x}_{s,s'}=w_0(s,s')$ is a solution for this set of equations.
\end{example}

In the following we propose to use the symbolic game graph to solve
the infinite game. By doing so, we obtain a polynomial time procedure.
We provide an inductive
construction of  vertex regions $U^i$ (for $i \in \mathbb{N}$) containing the
collection of vertices from which the Refuter has a strategy
for reaching the error state with probability greater than $0$ in at
most $i$ steps.

  Let
  $\SymbG_{A,A'} = ( V^{\SymbG},  E^{\SymbG}, V^{\SymbG}_\Refuter, V^{\SymbG}_\Verifier, V^{\SymbG}_\Probabilistic, \InitVertexSG )$
  be a symbolic game graph for PTSs $A$ and $A'$.
  Define $U = \bigcup_{i \geq 0} U^i$ where, for all $i\geq0$,
  \begin{align}
    U^0 = {} & \{\ErrorSt\} 
    & U^{i+1} = {}
    & \textstyle 
    \{v' \mid v' \in V^\SymbG_\Refuter \wedge \post^\SymbG(v') \cap (\bigcup_{j \leq i}U^j) \neq \emptyset \} \cup {}
    \label{def:U-eq}\\
    & & & \textstyle 
    \{v' \mid v' \in V^\SymbG_\Verifier \wedge \post^\SymbG(v') \subseteq \bigcup_{j\leq i}U^j \}  \cup {}
    \notag\\
    & & & \textstyle 
    \{v' \mid v' \in V^\SymbG_\Probabilistic \wedge \Eq(v', \post^\SymbG(v') \cap (\bigcup_{j \leq i}U^j)) \text{ has no solution}\}
    \notag
  \end{align}
The first line in $U^{i+1}$ corresponds to the Refuter and adds a
vertex if some successor is in some previous level $U^j$.  The second
line corresponds to the Verifier and adds a vertex if all its
successors lie in some previous $U^j$.  The last line corresponds to
the probabilistic player.
Notice that, if $\Eq(v', \post(v') \cap U^{i})$ has no solution, then every possible coupling will inevitably lead with some
probability to a ``losing'' state of a smaller level since, in
particular, equation
$\sum_{(s, \mhyphen, s', \mhyphen, \mhyphen, \Refuter) \in (\post(v') {\cap} U^{i})} x_{s,s'} = 0$
cannot be satisfied.

The following theorem provides an algorithm to decide
the stochastic masking game.
\begin{theorem}\label{thm:deciding:the:bisim:game}
  Let $\StochG_{A,A'}$ be a stochastic game graph for PTSs $A$ and
  $A'$, and let $\SymbG_{A,A'}$ be the corresponding symbolic game
  graph.
  Then, the Verifier has a sure (or almost-sure) winning strategy in
  $\StochG_{A,A'}$ for $\neg\Diamond\ErrorSt$ if and only if
  $v^\SymbG_0 \notin U$.
\end{theorem}

Theorems~\ref{thm:wingameStrat} and~\ref{thm:deciding:the:bisim:game} provide an alternative
algorithm to decide whether there is a probabilistic masking
simulation between $A$ and $A'$.
This can be done in polynomial time, since $\Eq(v,C)$ can be solved in
polynomial time (e.g, using linear programming) and the number of
iterations to construct $U$ is bounded by $|V^\SymbG|$.  Since
$V^\SymbG$ linearly depends on the transitions of the involved PTSs,
the complexity is in
$O(\textit{Poly}({|{\xrightarrow{}}|}\cdot{|{\xrightarrowprime{}}|)})$.

\section{Quantifying Fault Tolerance} \label{sec:probAlmostSure}

Probabilistic masking simulation determines whether a fault-tolerant
implementation is able to completely mask faults.  However, in
practice, this kind of masking fault-tolerance is uncommon. Usually,
fault-tolerant systems are able to mask a number of faults before
exhibiting a failure.
In this section we extend the game theory presented
above to provide a measure
for the system effectiveness on masking faults.
To do this, we extend the stochastic masking game with a
quantitative objective function.
The expected value of this function collects the (weighted) ``milestones''
that the fault-tolerant implementation is expected to cross before
failing.  A milestone is any interesting event that may occur during
a system execution.  For instance, a milestone may be 
the successful masking of a fault.
In this case, the measure will reflect the number of faults that are
tolerated by the system before crashing. Another milestone may be
successful acknowledgments in a transmission protocol. This
measures the expected number of chunks that the protocol is able to
transfer before failing. 
%
Thus, milestones are some designated action labels on the
implementation model and, as they may reflect different events, their
value may depend on the importance of such events.

\begin{definition}%
  Let $A' =( S', \SigmaF, \rightarrow', s_0' )$ be a PTS modeling an
  implementation.
  A \emph{milestone} is a function $\milestones:\SigmaF\to\Nat_0$.
\end{definition}
%
Given a milestone $\milestones$ for $A'$, the reward
$\mreward^\StochG$ on
$\StochG_{A,A'} = (V^\StochG, E^\StochG, V^\StochG_\Refuter, V^\StochG_\Verifier, V^\StochG_\Probabilistic, \InitVertex, \delta^\StochG)$ is defined by
$\mreward^\StochG(v) = \milestones(\sigma)$ if $v\in V^\StochG_\Verifier$ and
$\pr{1}{v}\in\{\sigma^1,\sigma^2\}$; otherwise,
$\mreward^\StochG(v) = 0$.
Function $\mreward^\StochG$ collects milestones (when available) only
once for each round of the game. This can be done only at Verifier's
vertices since they are the only ones that save the label that it is
being played in the round.
The \emph{masking payoff function} is then defined by
$\FMask(\rho) = \lim_{n \rightarrow \infty} (\sum^{n}_{i=0} \mreward^\StochG(\rho_i))$.
%
Therefore, the payoff function $\FMask$ represents the total of weighted
milestones that a fault-tolerant implementation is able to achieve
until an error state is reached.
This type of payoff functions are usually called \emph{total rewards}
in the literature.
One may think of this as a game played by the fault-tolerance built-in
mechanism and a (malicious) player that chooses the way in which
faults occur. In this game, the Verifier is the maximizer (she intends
to obtain as many milestones as possible) and the Refuter is the
minimizer (she intends to prevent the Verifier from collecting
milestones).

Thus, the game aims to optimize 
$\Expect{\strat{\Verifier}}{\strat{\Refuter}}_{\StochG, \InitVertex}[\FMask]$, i.e., the expected value of random variable $\FMask$.
%
One technical issue with total rewards objectives is that the game value may be not well-defined in $\mathbb{R}$.  For instance, there could be plays not reaching an end state wherein the players collect an infinite amount of rewards. A usual condition  for ensuring that the game value is well-defined is that of \emph{almost-surely stopping},  i.e.,  the game has to  reach a
sink vertex with probability~1, for every pair of strategies~\cite{FilarV96}.  In~\cite{DBLP:conf/cav/CastroDDP22}, we have generalized this condition to that of \emph{almost-surely stopping under fairness}, that is,  the error state
$\ErrorSt$ is reached with probability~1 provided the Refuter plays
fair.  In this case the games are well-defined in $\mathbb{R}$ and determined.  In simple words,  determination means that the knowledge of the opponent's strategy gives no benefit for the players.
%
%

It is worth noting that fairness is necessary to prevent the Refuter from stalling the
game. 
%
For instance, consider 
Example~\ref{example:memory} and the stochastic masking game
between the nominal and faulty models of
Figs.~\ref{fig:exam1MemCell:nom} and~\ref{fig:exam1MemCell:ft} (omitting the red part).
One would expect that the game leads to a failure with probability~1.
%
%
However, the Refuter has strategies to avoid $\ErrorSt$ with positive
probability.
For instance, the Refuter may
always play the reading action forcing the Verifier to mimic it
forever and hence making the probability of reaching the error
equals~0.
By doing this, the Refuter stalls the game,  forbidding  progress and
hence avoiding the occurrence of the fault.
Clearly, this is against  the intuitive behavior of faults which
one expects will eventually occur if waiting long enough.
The assumption of fairness over  Refuter plays  rules out this counter-intuitive behavior of the Refuter.
Roughly speaking,  a Refuter's fair play is one in which the Refuter
commits to follow a strong fair pattern, i.e., that includes
infinitely often any transition that is enabled infinitely often.  Then, a
fair strategy for the Refuter is a strategy that always measures 1 on
the set of all the Refuter's fair plays, regardless of the strategy of
the Verifier.
The definitions below follow the style
in~\cite{DBLP:journals/dc/BaierK98,BaierK08,DBLP:conf/cav/CastroDDP22}.


%
%
\begin{definition}
  Given a masking game
  $\StochG_{A,A'} = (V^\StochG, E^\StochG, V^\StochG_\Refuter, V^\StochG_\Verifier, V^\StochG_\Probabilistic, \InitVertex, \delta^\StochG)$,
  the \emph{set of all Refuter's fair plays} is defined by
  $ 
	\RFP = \{ \rho \in \Omega \mid v \in \inf(\rho) \cap V^\StochG_\Refuter \Rightarrow \post(v) \subseteq \inf(\rho) \}
  $.
  A Refuter strategy $\strat{\Refuter}$ is said to
  be \emph{almost-sure fair} iff, for every Verifier's strategy
  $\strat{\Verifier}$,
  $\Prob{\strat{\Refuter}}{\strat{\Verifier}}_{\StochG,\InitVertex}(\RFP) = 1$.
  %
  We let $\FairStrats{\Refuter}$ denote the set of all fair strategies
  for the Refuter.
\end{definition}

Under this concept, the stochastic masking game is almost-sure failing
under fairness if for every Verifier's strategy and every Refuter's
fair strategy, the game leads to an error with probability~1.  This is
formally defined as follows.

\begin{definition}
  Let $A$ and $A'$ be two PTSs.  We say that the stochastic masking
  game $\StochG_{A,A'}$ is \emph{almost-sure failing under fairness}
  iff, for every strategy
  $\strat{\Verifier} \in \Strategies{\Verifier}$ and any
  fair strategy
  $\strat{\Refuter} \in \FairStrats{\Refuter}$,
  $\Prob{\strat{\Verifier}}{\strat{\Refuter}}_{\StochG, \InitVertex}(\Diamond \ErrorSt)=1$.
\end{definition}

Interestingly, under the strong fairness assumption,
the determinacy of games is preserved for finite stochastic
games~\cite{DBLP:conf/cav/CastroDDP22}.  The rest of the section is
precisely devoted to bring our setting to the framework
of~\cite{DBLP:conf/cav/CastroDDP22} and thus provide an algorithmic
solution.

A strategy $\strat{i}$, $i\in\{\Refuter,\Verifier\}$, is
\emph{semi-Markov} if for every
$\hat{\rho},\hat{\rho}'\in(V^\StochG)^*$ and $v\in V^\StochG_i$,
$|\hat{\rho}|=|\hat{\rho}'|$ implies
$\strat{i}(\hat{\rho}v)=\strat{i}(\hat{\rho}'v)$, that is, the
decisions of $\strat{i}$ depend only on the length of the run and its
last state.  Thus, we write $\strat{i}(n,v)$ instead of
$\strat{i}(\hat{\rho}v)$ if $|\hat{\rho}|=n$. 
Let $\SemiMarkovStrats{i}$ denote the set of all semi-Markov
strategies for Player $i$ and $\SemiMarkovFairStrats{i}$ the set of
all its \emph{fair} semi-Markov strategies.

The next lemma states that, if the Refuter plays a semi-Markov
strategy, the Verifier achieves equal results regardless whether she plays an
arbitrary strategy or limits to playing only semi-Markov strategies.
The proof resembles that of~\cite[Lemma 2]{DBLP:conf/cav/CastroDDP22}
taking care of the fact that the set of vertices of the stochastic masking
game is uncountable.  Since probabilities are anyway discrete, this
is not a major technical issue,
but it deserves attention in the proof.

\begin{lemma}\label{lm:semmimarkov2}
  Let $\StochG_{A,A'}$ be a stochastic masking game graph and let
  $\strat{\Refuter} \in \SemiMarkovStrats{\Refuter}$ be a semi-Markov
  strategy.
  Then, for any $\strat{\Verifier} \in \Strategies{\Verifier}$,
  there is a semi-Markov strategy
  $\starredstrat{\Verifier} \in \SemiMarkovStrats{\Verifier}$
  such that
  $\Expect{\strat{\Verifier}}{\strat{\Refuter}}_{\StochG, v}[\FMask] =
   \Expect{\starredstrat{\Verifier}}{\strat{\Refuter}}_{\StochG, v}[\FMask]$.
\end{lemma}



A Verifyier strategy
$\starredstrat{\Verifier}\in\Strategies{\Verifier}$ is \emph{extreme} if it only moves to probabilistic
vertices containing couplings that are on the polytope vertices, that is, if for all
$\hat{\rho}\in(V^\StochG)^*{\times}V^\StochG_\Verifier$,
$\starredstrat{\Verifier}(\hat{\rho})((s,\mhyphen,s',\mu,\mu',w,\Probabilistic))>0$
implies that $w\in\vertices{\couplings{\mu}{\mu'}}$.
Let $\XSemiMarkovStrats{\Verifier}$ be the set of all extreme
semi-Markov strategies for the Verifier.

Lemma~\ref{lm:semmimarkov2} can be strengthened.  Thus, 
if the Refuter plays a semi-Markov strategy, the Verifier can achieve
the same result as the general case by  restricting herself to play  only extreme
semi-Markov strategies.

\begin{lemma}\label{lm:smstrat:to:smstrat:on:vertices}%
  Let $\StochG_{A,A'}$ be a stochastic masking game graph and let
  $\strat{\Refuter} \in \SemiMarkovStrats{\Refuter}$ be a semi-Markov
  strategy.
  Then, for any $\strat{\Verifier} \in \SemiMarkovStrats{\Verifier}$,
  there is an extreme semi-Markov strategy
  $\starredstrat{\Verifier} \in \XSemiMarkovStrats{\Verifier}$ such that
  for all $v\in V^\StochG_\Refuter$,
  $\Expect{\strat{\Verifier}}{\strat{\Refuter}}_{\StochG,v}[\FMask]
  = \Expect{\starredstrat{\Verifier}}{\strat{\Refuter}}_{\StochG,v}[\FMask]$.
\end{lemma}

The key of the proof of Lemma~\ref{lm:smstrat:to:smstrat:on:vertices}
lies on the construction of $\starredstrat{\Verifier}$ which is
defined so that, for every $n\in\Nat$ and $v_1\in V^\StochG_\Verifier$,
the probabilistic decision made by $\starredstrat{\Verifier}(n,v_1)$
corresponds to a proper composition of the probabilistic decisions of
$\strat{\Verifier}(n,v_1)$, and each convex combination of vertex
couplings that define the coupling within each probabilistic successor
$v_2\in\support{\strat{\Verifier}(n,v_1)}$.
%

Notice that, by traveling only through probabilistic vertices on
$\StochG_{A,A'}$ that are defined by vertex couplings, only a finite
number of the game vertices are touched when the Verifier uses extreme
strategies.
Thus, we let the stochastic game graph $\StochH_{A,A'}$ be the
\emph{vertex snippet of $\StochG_{A,A'}$} and define it to be the same
as $\StochG_{A,A'}$ only that probabilistic vertices are limited to
those that contain couplings in the vertices of the polytope, that is,
\begin{align*}
  V^{\StochH}_\Probabilistic = \
  & \{ (s, \mhyphen, s', \mu, \mu', w, \Probabilistic) \mid
       s \in S \wedge s' \in S' 
       \wedge w \in \vertices{\couplings{\mu}{\mu'}} \wedge {}\\
       & \phantom{\{ (s, \mhyphen, s', \mu, \mu', w, \Probabilistic) \mid {}}
       \exists \sigma{\in}\SigmaF :
       {({s\xrightarrow{\sigma} \mu} \vee (\sigma{\in}\faults \wedge \mu=\Dirac_s))}
       \wedge {s'\xrightarrowprime{\sigma}\mu'}\}.
\end{align*}
The rest of the elements of $\StochH_{A,A'}$ are defined by
properly restricting the domain of the respective components in
$\StochG_{A,A'}$.  Notice that $\StochH_{A,A'}$ is finite.

Now observe that, if the Verifier semi-Markov strategies are considered
as functions with domain in $(\Nat\times V^\StochG_\Verifier)$, then
the set of all extreme semi-Markov strategies in $\StochG_{A,A'}$
corresponds to the set of all semi-Markov strategies of
$\StochH_{A,A'}$. That is: 
$\XSemiMarkovStrats{\Verifier,\StochG}=\SemiMarkovStrats{\Verifier,\StochH}$,
where subscripts $\StochG$ and $\StochH$ indicate whether the strategies belong to $\StochG_{A,A'}$ or
$\StochH_{A,A'}$, respectively.
Similarly,  the same holds for the set of all extreme deterministic
memoryless strategies, that is: 
$\XDetMemorylessStrats{\Verifier,\StochG}=\DetMemorylessStrats{\Verifier,\StochH}$.
Given the fact that $V^\StochG_\Refuter=V^{\StochH}_\Refuter$
and $V^\StochG_\Verifier=V^{\StochH}_\Verifier$, the set of all
Refuter's deterministic memoryless fair strategies are the same
in both game graphs,
i.e., $\DetMemorylessFairStrats{\Refuter,\StochG}=\DetMemorylessFairStrats{\Refuter,\StochH}$.
The following proposition follows directly from these observations.

\begin{proposition}\label{prop:popurri}%
  Let $\StochG_{A,A'}$ be a stochastic game graph and
  $\StochH_{A,A'}$ its vertex snippet.
  Then, for all $v\in V^\StochG_\Refuter$($=V^{\StochH}_\Refuter$),
  we have:
  \begin{enumerate}
  \item\label{prop:popurri:iii}%
    for all
    $\strat{\Refuter}\in\DetMemorylessFairStrats{\Refuter,\StochG}$,($=\DetMemorylessFairStrats{\Refuter,\StochH}$),
    $\sup_{\strat{\Verifier}\in\XSemiMarkovStrats{\Verifier,\StochG}}\Expect{\strat{\Verifier}}{\strat{\Refuter}}_{\StochG,v}[\FMask] =
    \sup_{\strat{\Verifier}\in\SemiMarkovStrats{\Verifier,\StochH}}\Expect{\strat{\Verifier}}{\strat{\Refuter}}_{\StochH,v}[\FMask]$; and
  \item\label{prop:popurri:v}%
    for all
    $\strat{\Verifier}\in\XDetMemorylessStrats{\Verifier,\StochG}$($=\DetMemorylessStrats{\Verifier,\StochH}$),
    $\inf_{\strat{\Refuter}\in\FairStrats{\Refuter,\StochG}}\Expect{\strat{\Verifier}}{\strat{\Refuter}}_{\StochG,v}[\FMask] =
    \inf_{\strat{\Refuter}\in\FairStrats{\Refuter,\StochH}}\Expect{\strat{\Verifier}}{\strat{\Refuter}}_{\StochH,v}[\FMask]$.
  \end{enumerate}
\end{proposition}

The following theorem not only states that the game for optimizing
the expected value of the masking payoff function is determined, but
it also guarantees that it can be solved using the finite vertex snippet
of the stochastic game subgraph.

\begin{theorem}\label{theorem:determinate:memoryless:and:finite}%
  Let $\StochG_{A,A'}$ be a stochastic game graph whose vertex snippet
  $\StochH_{A,A'}$ is almost-sure failing under fairness.
  Then, for all $v\in V^\StochG_\Refuter$($=V^{\StochH}_\Refuter$),
  %
  \begin{multline*}
    \inf_{\strat{\Refuter}\in\FairStrats{\Refuter,\StochG}}\sup_{\strat{\Verifier}\in\Strategies{\Verifier,\StochG}}\Expect{\strat{\Verifier}}{\strat{\Refuter}}_{\StochG,v}[\FMask]
    =
    \inf_{\strat{\Refuter}\in\DetMemorylessFairStrats{\Refuter,\StochH}}\sup_{\strat{\Verifier}\in\DetMemorylessStrats{\Verifier,\StochH}}\Expect{\strat{\Verifier}}{\strat{\Refuter}}_{\StochH,v}[\FMask]
    \\
    =
    \sup_{\strat{\Verifier}\in\DetMemorylessStrats{\Verifier,\StochH}}\inf_{\strat{\Refuter}\in\DetMemorylessFairStrats{\Refuter,\StochH}}\Expect{\strat{\Verifier}}{\strat{\Refuter}}_{\StochH,v}[\FMask]
    =
    \sup_{\strat{\Verifier}\in\Strategies{\Verifier,\StochG}}\inf_{\strat{\Refuter}\in\FairStrats{\Refuter,\StochG}}\Expect{\strat{\Verifier}}{\strat{\Refuter}}_{\StochG,v}[\FMask].
  \end{multline*}
\end{theorem}
\begin{proof}
  We first recall that the almost-sure failing under fairness property
  is equivalent to the stopping under fairness property
  in~\cite{DBLP:conf/cav/CastroDDP22}.  That is why we can safely
  apply the results from~\cite{DBLP:conf/cav/CastroDDP22} on
  $\StochH_{A,A'}$ in the calculations below.
  \begin{align*}
    \textstyle
    \inf_{\strat{\Refuter}\in\FairStrats{\Refuter,\StochG}}\sup_{\strat{\Verifier}\in\Strategies{\Verifier,\StochG}}\Expect{\strat{\Verifier}}{\strat{\Refuter}}_{\StochG,v}[\FMask]
    & \textstyle {} \leq
    \inf_{\strat{\Refuter}\in\DetMemorylessFairStrats{\Refuter,\StochG}}\sup_{\strat{\Verifier}\in\Strategies{\Verifier,\StochG}}\Expect{\strat{\Verifier}}{\strat{\Refuter}}_{\StochG,v}[\FMask]
    & \text{($\DetMemorylessFairStrats{\Refuter,\StochG}\subseteq\FairStrats{\Refuter,\StochG}$)} &  \tag{$\star$}\\
    & \textstyle {} =
    \inf_{\strat{\Refuter}\in\DetMemorylessFairStrats{\Refuter,\StochG}}\sup_{\strat{\Verifier}\in\SemiMarkovStrats{\Verifier,\StochG}}\Expect{\strat{\Verifier}}{\strat{\Refuter}}_{\StochG,v}[\FMask]
    & \text{(by Lemma~\ref{lm:semmimarkov2})} & \notag\\
    & \textstyle {} =
    \inf_{\strat{\Refuter}\in\DetMemorylessFairStrats{\Refuter,\StochG}}\sup_{\strat{\Verifier}\in\XSemiMarkovStrats{\Verifier,\StochG}}\Expect{\strat{\Verifier}}{\strat{\Refuter}}_{\StochG,v}[\FMask]
    & \text{(by Lemma~\ref{lm:smstrat:to:smstrat:on:vertices})} & \notag\\
    & \textstyle {} =
    \inf_{\strat{\Refuter}\in\DetMemorylessFairStrats{\Refuter,\StochH}}\sup_{\strat{\Verifier}\in\SemiMarkovStrats{\Verifier,\StochH}}\Expect{\strat{\Verifier}}{\strat{\Refuter}}_{\StochH,v}[\FMask]
    & \text{(by Prop.~\ref{prop:popurri}.\ref{prop:popurri:iii})} & \notag\\
    & \textstyle {} \leq
    \inf_{\strat{\Refuter}\in\DetMemorylessFairStrats{\Refuter,\StochH}}\sup_{\strat{\Verifier}\in\DetMemorylessStrats{\Verifier,\StochH}}\Expect{\strat{\Verifier}}{\strat{\Refuter}}_{\StochH,v}[\FMask]
    & \text{(by~\cite[Thm.~5]{DBLP:conf/cav/CastroDDP22})} & \tag{$\star$}\\
    & \textstyle {} =
    \sup_{\strat{\Verifier}\in\DetMemorylessStrats{\Verifier,\StochH}}\inf_{\strat{\Refuter}\in\DetMemorylessFairStrats{\Refuter,\StochH}}\Expect{\strat{\Verifier}}{\strat{\Refuter}}_{\StochH,v}[\FMask]
    & \text{(by~\cite[Thm.~5]{DBLP:conf/cav/CastroDDP22})} & \tag{$\star$}\\
    & \textstyle {} =
    \sup_{\strat{\Verifier}\in\DetMemorylessStrats{\Verifier,\StochH}}\inf_{\strat{\Refuter}\in\FairStrats{\Refuter,\StochH}}\Expect{\strat{\Verifier}}{\strat{\Refuter}}_{\StochH,v}[\FMask]
    & \text{(by~\cite[Lemma~6]{DBLP:conf/cav/CastroDDP22})} & \notag\\
    & \textstyle {} =
    \sup_{\strat{\Verifier}\in\XDetMemorylessStrats{\Verifier,\StochG}}\inf_{\strat{\Refuter}\in\FairStrats{\Refuter,\StochG}}\Expect{\strat{\Verifier}}{\strat{\Refuter}}_{\StochG,v}[\FMask]
    & \text{(by Prop.~\ref{prop:popurri}.\ref{prop:popurri:v})} & \notag\\
    & \textstyle {} \leq
    \sup_{\strat{\Verifier}\in\Strategies{\Verifier,\StochG}}\inf_{\strat{\Refuter}\in\FairStrats{\Refuter,\StochG}}\Expect{\strat{\Verifier}}{\strat{\Refuter}}_{\StochG,v}[\FMask]
    & \text{($\XDetMemorylessStrats{\Verifier,\StochG}\subseteq\Strategies{\Verifier,\StochG}$)} & \tag{$\star$}\\
    & \textstyle {} \leq
    \inf_{\strat{\Refuter}\in\FairStrats{\Refuter,\StochG}}\sup_{\strat{\Verifier}\in\Strategies{\Verifier,\StochG}}\Expect{\strat{\Verifier}}{\strat{\Refuter}}_{\StochG,v}[\FMask]
    & \text{(prop. $\inf$/$\sup$)} & \notag
  \end{align*}
  Formulas marked with ($\star$) are those in the statement
  of the theorem and, because the first and last formulas
  are the same, all of them are equal.
\end{proof}

Theorem~\ref{theorem:determinate:memoryless:and:finite} guarantees
that the stochastic masking game can be solved through its finite
vertex snippet using the algorithm proposed in~\cite{DBLP:conf/cav/CastroDDP22}.
The next theorem uses this fact to provide a set of
Bellman equations based on the symbolic game graph whose greatest
fixpoint solution is the solution of the original stochastic masking
game.

\begin{theorem}\label{thm:algoritmic:solution:of:G}%
  Let $\StochG_{A,A'}$ be a stochastic masking game graph whose vertex
  snippet is almost-sure failing under fairness and let $\milestones$
  be a milestone for $A'$.
  Let $\SymbG_{A,A'}$ be the corresponding symbolic game graph.
  Let $\nu\Bellman$ be the greatest fixpoint of the functional
  $\Bellman$ defined, for all $v\in V^\SymbG$, as follows:{\small%
  \[
  \Bellman(f)(v) =
  \begin{cases}
    \min \big(\upperbound,  \max_{w \in \vertices{\couplings{\pr{3}{v}}{\pr{4}{v}}}} \sum_{v' \in \post(v)} w(\pr{0}{v'},\pr{2}{v'})  f(v') \big)
    & \text{if $v \in V^{\SymbG}_\Probabilistic$}
    \\
    \min \left( \upperbound, \mreward^\SymbG(v) + \max \left\{f(v') \mid v' \in \post(v) \right\} \right)
    & \text{if $v \in  V^{\SymbG}_\Verifier$}
    \\
    \min \left( \upperbound,  \min \left\{f(v') \mid v' \in \post(v) \right) \right\}
    & \text{if $v \in {V^{\SymbG}_\Refuter{\setminus}\{\ErrorSt\}}$}
    \\
    0
    & \text{if $v=\ErrorSt$}
  \end{cases}
  \]}%
  %
  where 
  $\pr{i}{v}$ is the $i$-th coordinate of $v$ ($i\geq 0$),
  $\mreward^\SymbG(\pr{0}{v},\pr{1}{v},\pr{2}{v},\pr{3}{v},\pr{4}{v},\pr{6}{v}) = \mreward^\StochG(v)$
  for every $v\in V^\StochG$, and
  $\upperbound\in\Reals$ such that 
  $\upperbound \geq \inf_{\strat{\Refuter} \in \DetMemorylessFairStrats{\Refuter}} \sup_{\strat{\Verifier} \in \DetMemorylessStrats{\Verifier}} \Expect{\strat{\Verifier}}{\strat{\Refuter}}_{\StochG_{A,A'}, v}[\FMask]$,
  for every $v\in V^\StochG$.
  Then, the value of the game $\StochG_{A,A'}$ at its initial state is
  equal to $\nu\Bellman(\InitVertexSG)$.
\end{theorem}
Constant $\upperbound$ is an upper bound needed so Knaster-Tarski applies on the
complete lattice $[0,\upperbound]^V$~\cite{DBLP:conf/cav/CastroDDP22}.

Notice that Theorems~\ref{theorem:determinate:memoryless:and:finite}
and~\ref{thm:algoritmic:solution:of:G} only require $\StochH_{A,A'}$
to be almost-sure failing under fairness, and if $\StochG_{A,A'}$ is
almost-sure failing under fairness, necessarily so is
$\StochH_{A,A'}$, which makes the theorems stronger.
Nonetheless, one would expect that also if $\StochH_{A,A'}$ is
almost-sure failing under fairness, so is $\StochG_{A,A'}$.
That is, we would like that
$\inf_{\strat{\Verifier}\in\Strategies{\Verifier},\strat{\Refuter}\in\FairStrats{\Refuter}}\Prob{\strat{\Verifier}}{\strat{\Refuter}}_{\StochG,\InitVertex}(\Diamond\ErrorSt)=1$
if and only if
$\inf_{\strat{\Verifier}\in\Strategies{\Verifier},\strat{\Refuter}\in\FairStrats{\Refuter}}\Prob{\strat{\Verifier}}{\strat{\Refuter}}_{\StochH,\InitVertexH}(\Diamond\ErrorSt)=1$.
Unfortunately we were not able to prove this equivalence, and the most
we know (thanks to variants of Lemmas~\ref{lm:semmimarkov2}
and~\ref{lm:smstrat:to:smstrat:on:vertices}) is that
\sloppy $\inf_{\strat{\Verifier}\in\Strategies{\Verifier},\strat{\Refuter}\in\FairStrats{\Refuter}}\Prob{\strat{\Verifier}}{\strat{\Refuter}}_{\StochH,\InitVertexH}(\Diamond\ErrorSt)=1$
implies both
$\inf_{\strat{\Verifier}\in\Strategies{\Verifier},\strat{\Refuter}\in\SemiMarkovFairStrats{\Refuter}}\Prob{\strat{\Verifier}}{\strat{\Refuter}}_{\StochG,\InitVertex}(\Diamond\ErrorSt)=1$
and
$\inf_{\strat{\Verifier}\in\SemiMarkovStrats{\Verifier},\strat{\Refuter}\in\FairStrats{\Refuter}}\Prob{\strat{\Verifier}}{\strat{\Refuter}}_{\StochG,\InitVertex}(\Diamond\ErrorSt)=1$,
that is, at least one of the set of strategies needs to be restricted
to the semi-Markov ones. 

Since $\StochH_{A,A'}$ is finite, it can be checked whether it is
almost-sure failing under fairness by using directly the algorithm
proposed in~\cite[Theorem~3]{DBLP:conf/cav/CastroDDP22}.  However, we
could alternatively check it avoiding the explosion introduced by the
vertex couplings through the symbolic game graph.
Thus,  we define the predecessor sets in $\SymbG_{A,A'}$ for a given set
$C$ of symbolic vertices, as follows:
{\small
\begin{align*}
  \SymbEFairpre(C) = {}
  & \{ v \in V^\SymbG \mid {\post(v)\cap C \neq \emptyset} \}\\
  \SymbAFairpre(C) = {}
  & \{ v \in V^\SymbG_\Verifier \mid {\post(v)\subseteq C} \}
    \cup \{ v \in V^\SymbG_\Refuter \mid {\post(v)\cap C \neq \emptyset} \} \\
  & \cup \{ v \in V^\SymbG_\Probabilistic \mid \Eq(v,C) \text{ has no solution }\}
\end{align*}}%
$\SymbEFairpre(C)$ collects all vertices $v$ for which there is a
coupling that leads to a vertex $v'$ in $C$, and do so by simply using
the edge $E^\SymbG$ (through $\post$) even for the probabilistic
vertices.
The definition of $\SymbAFairpre(C)$ is more assorted.
The first set collects all the Verifier vertices $v$ that inevitably
lead to $C$.
The second set collects all Refuter vertices $v$ that leads to some
state in $C$ (since the Refuter is fair, any successor of $v$ will
eventually be taken).
The last set collects all probabilistic vertices $v$ for which there
is no coupling ``avoiding'' $C$.  This is encoded by checking that
$\Eq(v,C)$ cannot be solved, since a coupling solving $\Eq(v,C)$
defines a probabilistic transition that avoids $C$ with probability 1.


The next theorem provides an algorithm to check whether a vertex
snippet is almost-sure failing under fairness using $\SymbEFairpre$
and $\SymbAFairpre$.

\begin{theorem}\label{theo:decide-stopping}%
  The vertex snippet $\StochH_{A,A'}$ of the stochastic masking game 
  $\StochG_{A,A'}$ is almost-sure failing under fairness if and only if
  $\InitVertexSG \in V^\SymbG \setminus {\SymbEFairpre}^*(V^\SymbG \setminus {\SymbAFairpre}^*(\{ \ErrorSt \}))$,
  where $\InitVertexSG$ is the initial state of $\SymbG_{A,A'}$ (the
  symbolic version of $\StochG_{A,A'}$) and $V^\SymbG$ is the sets of
  vertices of $\SymbG_{A,A'}$.
\end{theorem}
As $\Eq(v,C)$ can be computed in polynomial time, so do
$\SymbEFairpre(C)$ and $\SymbAFairpre(C)$.  As a consequence, the
problem of deciding whether a vertex snippet $\StochH_{A,A'}$ is
almost-sure failing under fairness is polynomial on the sizes of $A$
and $A'$.




\section{Related Work}\label{sec:related-work}

Since our metric is a bisimulation-based notion aimed at quantifying
how robust a masking fault tolerant algorithm is, the idea of
approximate bisimulation immediately shows up.
In this category it is worth mentioning 
$\epsilon$-bisimulations~\cite{DBLP:conf/ifip2/GiacaloneJS90,DBLP:conf/lics/DesharnaisJGP02},
in which related  states that imitate each other do not differ more than an
$\epsilon\in[0,1]$ on the probabilistic value.  Therefore
$\epsilon$-bisimulations are not able to accumulate the difference produced in each step.  So,  these  relations cannot measure to what extent faults can be
tolerated over time.
The principle of (1-bounded) bisimulation
metrics~\cite{DBLP:conf/qest/DesharnaisLT08,DBLP:conf/birthday/BreugelW14}
is different as they aim to quantify the similarity of whole models
rather than single steps.  Nonetheless, if the models inevitably differ
(as it is the case of  almost-sure failing systems) the
metric always equals 1 (maximum difference), which again cannot
measure how long faults are tolerated.
Instead, bisimulation metrics with
discount~\cite{DBLP:conf/qest/DesharnaisLT08,DBLP:conf/birthday/BreugelW14}
do give an idea of robustness since the discount factor inversely
weights how distant in a trace the difference between the models is
eventually witnessed.  However, these metrics only provide a relative
value (smaller values mean more robust) and cannot focus on particular
events as our metric does.
In any case, all these notions have been characterized by games
which served as a base for algorithmic
solutions~\cite{DBLP:conf/lics/DesharnaisJGP02,DBLP:conf/birthday/BreugelW14,BacciBLMTB19,DBLP:journals/jcss/TangB20}.
%
In~\cite{DBLP:conf/lics/DesharnaisJGP02} a \emph{non-stochastic} game
for $\epsilon$-bisimulation is provided where each round is divided in
five steps in which both Refuter and Verifier alternate twice.
Therein the difference is quantified independently in each step, so it
is easy to avoid the use of couplings.
%
Instead, the stochastic games for bisimulation
metrics~\cite{DBLP:conf/birthday/BreugelW14} are very much similar to
ours with the difference that the Verifier only chooses a vertex
coupling instead of any possible coupling as we do here, and it
considers only deterministic memoryless strategies.

In \cite{LanotteMT17} a weak simulation quasimetric is introduced and used to reason
about the evolution of \emph{gossip protocols} to compare
protocols with similar behavior up to a certain tolerance.
Though its purpose is close to ours,  the quasimetric suffers the
same problem as bisimulation metrics returning $1$ 
when comparing protocols with almost-sure failing implementations.

Metrics like \emph{Mean-Time To Failure} (MTTF)~\cite{ReliabilityBook} are normally used.
However, our framework is more general than such metrics since it is not limited to count time
 units  as other events may be set as milestones.  In addition, the computation of MTTF would normally require the identification of
 failure states in an ad hoc manner while we do this at a higher
 level of abstraction.

\section{Concluding remarks} \label{sec:finalDiscussions}

We presented a relation of masking fault-tolerance between
probabilistic transition systems and a corresponding stochastic game
characterization.  As the game could be infinite, we proposed an
alternative finite symbolic representation by means of which the game
can be solved in polynomial time.
We extended the game with quantitative objectives based on collecting
``milestones'' thus providing a way to quantify how good an
implementation is for masking faults.
%
We proved that the resulting game is determined and can be computed by
solving a collection of functional equations.  We also provided a
polynomial technique to decide whether a game is almost-sure failing
under fairness.
In this article we focused on the theoretical contribution.  We leave
as further work the description of the implementation of this idea.

Though it does not affect our result of determinacy nor the
algorithmic solution proposed here, it remains open to show whether it
holds that whenever the vertex snippet is almost-sure failing under
fairness so is the general stochastic masking game.
%
Also, notice that our solution is based on a strong version of
bisimulation.  A characterization based on probabilistic weak
bisimulation would facilitate the application of our approach to
complex systems.
\bibliographystyle{eptcs}
\bibliography{references-dblp}


\end{document}